\documentclass[11pt,fleqn]{article}
\usepackage{etex}
\usepackage{amssymb,latexsym,amsmath,amsfonts,graphicx, ebezier}
\usepackage{pictex,color}
\usepackage{pict2e}
\usepackage[textsize=footnotesize]{todonotes}
    \advance\textheight by \topskip

    \numberwithin{equation}{section}

\makeatletter
\def\eqalign#1{\null\vcenter{\def\\{\cr}\openup\jot\m@th
  \ialign{\strut$\displaystyle{##}$\hfil&$\displaystyle{{}##}$\hfil
      \crcr#1\crcr}}\,}
\makeatother

\def\beq{\begin{equation}}
\def\eeq{\end{equation}}

\newcommand{\be}{\begin{equation}}
\newcommand{\ee}{\end{equation}}

\def\Var{{\rm Var}}

    \def\Re{{\rm Re \,}}
    \def\Im{{\rm Im \,}}
    \def\Ai{{\rm Ai \,}}

    \def\bigO{{\cal O}}

    \def\P2n{{\rm P}_{{\rm II}}^{(n)}}

    \newtheorem{theorem}{Theorem}[section]
    
    \newtheorem{corollary}[theorem]{Corollary}
    \newtheorem{proposition}[theorem]{Proposition}
    
    \newtheorem{Definition}[theorem]{Definition}
    
    \newtheorem{Remark}[theorem]{Remark}
    \newenvironment{remark}{\begin{Remark}\rm}{\end{Remark}}
    \newtheorem{Example}[theorem]{Example}
    
    \newtheorem{Assumptions}[theorem]{Assumptions}

    \newenvironment{proof}%
    {\rm \trivlist \item[\hskip \labelsep{\bf Proof. }]}%
    {\hspace*{\fill}$\Box$\endtrivlist}
    {\rm \trivlist \item[\hskip \labelsep{\bf Proof}]}%
    {\hspace*{\fill}$\Box$\endtrivlist}

    \newcommand{\supp}{{\operatorname{supp}}}

\begin{document}
\title{Thinning and conditioning of the Circular Unitary Ensemble}
\author{Christophe Charlier\footnote{Institut de Recherche en Math\'ematique et Physique,  Universit\'e
catholique de Louvain, Chemin du Cyclotron 2, B-1348
Louvain-La-Neuve, BELGIUM}\ \ and Tom Claeys\footnotemark[\value{footnote}]}
\maketitle
%\tableofcontents

\begin{abstract}
We apply the operation of random independent thinning on the eigenvalues of $n\times n$ Haar distributed unitary random matrices. We study gap probabilities for the thinned eigenvalues, and we study the
statistics of the eigenvalues of random unitary matrices which are conditioned such that 
there are no thinned eigenvalues on a given arc of the unit circle. Various probabilistic quantities can be expressed  in terms of Toeplitz determinants and orthogonal polynomials on the unit circle, and we use these expressions to obtain asymptotics as $n\to\infty$.  \end{abstract}

\section{Introduction}\label{sec: intro}

Randomly incomplete spectra of random matrices were introduced by Bohigas and Pato in \cite{BohigasPato1, BohigasPato2}, motivated by problems in nuclear physics.
Two natural questions arose in those works: (1) what can we say about the eigenvalues of the incomplete spectrum, and (2) to what extent does the incomplete spectrum give us information about the complete spectrum? We investigate these questions in detail for Haar distributed unitary matrices in the limit where the matrix size becomes large. 

Thinning is a classical operation in the theory of point processes \cite{Daley, Kallenberg, Renyi}. Given a locally finite configuration of points, it consists of removing points following a deterministic or probabilistic rule. Starting with a given point process, one can in this way create other, thinned, point processes.
If one thins a determinantal point process (see e.g.\ \cite{Johansson, Soshnikov} for an overview of the theory of determinantal point processes) by independently removing each particle with a certain probability, it is remarkable, but straightforward to see, that the resulting process is still determinantal \cite{LMR}.

In this paper, we start from the determinantal point process on the unit circle in the complex plane defined by the eigenvalues of a random Haar distributed unitary random matrix, and we apply a random uniform independent thinning operation to it. This means that we remove each eigenvalue independently with a given probability  which does not depend on the position of the eigenvalue. One can think of the remaining eigenvalues as observed particles and the removed eigenvalues as unobserved particles.

On one hand, we will study the asymptotic behaviour of gap probabilities for the thinned spectrum as the size of the matrices tends to infinity. On the other hand, we will study the asymptotic statistical behaviour of all eigenvalues, given some knowledge about the thinned (observed) eigenvalues. More precisely, we will assume that a certain arc of the unit circle is free of thinned particles, and we will investigate what this condition tells us about the other eigenvalues.

The eigenvalues of Haar distributed unitary matrices are of particular interest because of their remarkable connection to zeros with large imaginary part of the Riemann $\zeta$-function on the critical line \cite{KeatingSnaith}. Spacings between thinned eigenvalues in this model were compared recently \cite{ForresterMays} with thinned data sets of zeros of the Riemann $\zeta$-function, showing accurate agreement.

\medskip

We begin with a more precise definition of the models which we study.

\paragraph{Thinned Circular Unitary Ensemble.}

An $n\times n$ matrix $U$ from the Circular Unitary Ensemble (CUE) is a random Haar distributed unitary $n\times n$ matrix.
We denote $e^{i\theta_1},\ldots, e^{i\theta_n}$ for the eigenvalues of $U$, with $\theta_j\in[0,2\pi)$ for $j=1,\ldots, n$.
The joint probability distribution of the eigenvalues is given by
\begin{equation}\label{jpdf CUE}
\frac{1}{(2\pi)^{n}n!}\prod_{1\leq j<k\leq n}|e^{i\theta_j}-e^{i\theta_k}|^2 \prod_{j=1}^nd\theta_j,\qquad \theta_1,\ldots, \theta_n\in [0,2\pi).
\end{equation}
This is a determinantal point process on the unit circle with correlation kernel \cite{Mehta}
\begin{equation}\label{CUE kernel}
K_n(e^{i\theta},e^{i\mu})=\frac{1}{2\pi}\sum_{k=0}^{n-1}e^{ik(\theta-\mu)}.
\end{equation}

Consider the operation of thinning the eigenvalues of an $n\times n$ CUE matrix, which consists of removing each eigenvalue $e^{i\theta_{1}},...,e^{i\theta_{n}}$ independently  with probability $s=1-p \in (0,1)$. We call $s$ the {\em removal probability} and $p$ the {\em retention probability}. We denote $e^{i\phi_{1}},...,e^{i\phi_{m}}$, $0\leq m \leq n$ for the retained eigenvalues, where the number of particles $m$ is now itself a random variable following the binomial distribution $B(n,p)$. We write 
\[\Theta=\{e^{i\theta_1},\ldots, e^{i\theta_n}\},\qquad \Phi=\{e^{i\phi_1},\ldots, e^{i\phi_m}\}\]
for the complete and thinned spectra.
On the formal level, one can define the thinned CUE by associating to each eigenvalue $e^{i\theta_k}$ an independent Bernoulli random variable taking the value $0$ if the eigenvalue is removed and taking the value $1$ if it is kept.

Thinned point processes interpolate in general between the original point process (for $p=1$) and a Poisson process (as $p\to 0$, after a proper re-scaling) with uncorrelated points \cite{Kallenberg}.
Since the eigenvalues of a CUE matrix are intrinsically repulsive by \eqref{jpdf CUE}, in our situation the retention probability $p$ can be thought of as a parameter measuring the repulsion between particles. The repulsion is strongest for $p=1$, becomes weaker as $p$ decreases, and disappears as $p\to 0$.

On the level of correlation kernels, random independent thinning of a general determinantal point process boils down to  multiplying the correlation kernel with the retention probability $p$, as was observed in \cite[Appendix A]{LMR}.
The correlation kernel for the particles $e^{i\phi_{1}},...,e^{i\phi_{m}}$ is in other words equal to $pK_n(e^{i\theta},e^{i\mu})$, with $K_n$ as in \eqref{CUE kernel}.

It is worth noting that the CUE kernel $K_n$, seen as the kernel of an integral operator acting on $L^2(S^1)$ with $S^1$ the unit circle, defines an orthogonal projection operator of rank $n$ onto the space of polynomials of degree $\leq n-1$. The integral operator associated to the kernel $pK_n$ with $p\in (0,1)$ is still of rank $n$, but not a projection. By the general theory of determinantal point processes \cite{Soshnikov}, this is consistent with the fact that the number of particles in the original process is equal to $n$, whereas it is a random number $\leq n$ in the thinned process.

\medskip

Although eigenvalue correlations (and large $n$ asymptotics thereof) in the thinned CUE are obtained directly from the non-thinned eigenvalue correlations, the question of computing large $n$ asymptotics for gap probabilities and, related to that, eigenvalue spacings, is less trivial. For finite $n$, the probability $\mathbb P\left(\Phi\subset\gamma\right)$ that all thinned eigenvalues lie on a measurable part $\gamma\subset S^1$ is equal to the Fredholm determinant $\det\left(1-p\left.K_n\right|_{\gamma^c}\right)$, where $\gamma^c=S^1\setminus\gamma$ and $\left.pK_n\right|_{\gamma^c}$ is the integral operator with kernel $pK_n$ acting on $L^2(\gamma^c)$.
In the special case where $\gamma^c$ is an arc of the unit circle with arclength $4\pi y/n$ with $y>0$ fixed, the gap probability $\mathbb P\left(\Phi\subset\gamma\right)$ converges as $n\to\infty$ to the Fredholm determinant $\det\left(1-p\left.K^{\mathrm{sin}}\right|_{[-y,y]}\right)$ associated to the sine kernel operator with kernel $pK^{\mathrm{sin}}(u,v)=p\frac{\sin\pi(u-v)}{\pi(u-v)}$. Large gap asymptotics as $y\to\infty$ and $p\to 1$ for this determinant were studied in detail in \cite{BDIK}.

\medskip

For general measurable subsets $\gamma$ of $S^1$, we can also express $\mathbb P\left(\Phi\subset\gamma\right)$ as a Toeplitz determinant in the following elementary way.
We can partition $[0,2\pi)^{n}$ as
\begin{equation}\label{partition}
[0,2\pi)^{n} = \bigsqcup_{k=0}^{n} A_{k}, \qquad A_{k} = \{ (\theta_{1},...,\theta_{n}) \in [0,2\pi)^{n} : \# (\Theta\cap\gamma^c)=k \}.
\end{equation} In this way $A_k$ represents the event that exactly $k$ eigenvalues belong to $\gamma^c$.
A given eigenvalue configuration $\Theta$ leads to a thinned configuration $\Phi$ which lies entirely in $\gamma$ with probability $s^k$, where $k=\#(\Theta\cap\gamma^c)$.
Therefore,
\begin{equation}\label{gap probability}
\mathbb P\left(\Phi\subset\gamma\right)
=\sum_{k=0}^n s^k\mathbb P\left(A_k\right)=\frac{1}{(2\pi)^{n}n!} \sum_{k=0}^{n} s^{k} \int_{A_{k}}\prod_{1\leq j<k\leq n}|e^{i\theta_j}-e^{i\theta_k}|^2 \prod_{j=1}^n d\theta_j.
\end{equation}
Writing $f$ for the piece-wise constant function 
\begin{equation}\label{def symbol}
f(z)=\begin{cases}1& \mbox{ for }z\in\gamma,\\
s& \mbox{ for }z\in\gamma^c,\end{cases}
\end{equation}
we can express the gap probability as a Toeplitz determinant with symbol $f$:
\begin{multline}\label{gap prob Toeplitz}
\mathbb P\left(\Phi\subset\gamma\right)
=\frac{1}{(2\pi)^{n}n!}  \int_{[0,2\pi]^n}\prod_{1\leq j<k\leq n}|e^{i\theta_j}-e^{i\theta_k}|^2 \prod_{j=1}^n f(e^{i\theta_j})d\theta_j\\
=\det\left(\frac{1}{2\pi}\int_0^{2\pi}f(e^{i\theta})e^{-i(k-j)\theta}d\theta\right)_{j,k=1,\ldots, n}=:D_n(f),
\end{multline}
where we used the standard Heine integral representation for Toeplitz determinants.
If $\gamma$ is an arc of arclength $2L$, $f$ given by \eqref{def symbol} has two jump discontinuities and is a special case of a symbol with two Fisher-Hartwig singularities. Throughout the paper, we will write $D_n(s,L)$ for $D_n(f)$ in this case.  The large $n$ asymptotics of Toeplitz determinants for symbols with Fisher-Hartwig singularities have a long history, notable results in this context have been obtained  e.g.\ in \cite{B, BS, DIK, DIK2, Ehrhardt, FH, W}.

If $L\in (0,\pi)$ and $s\in(0,1)$ are independent of $n$, large $n$ asymptotics for $D_n(s,L)$ are known, and we have \cite{DIK, Ehrhardt}
\begin{multline}\label{as FH}
\mathbb P\left(\Phi\subset\gamma\right)=D_n(s,L)=s^{n(1 - \frac{L}{\pi})}n^{\frac{(\log s)^2}{2\pi^2}}(2\sin L)^{\frac{(\log s)^{2}}{2\pi^{2}}} \\
\times\  G\left(1+\frac{\log s}{2\pi i}\right)^2G\left(1-\frac{\log s}{2\pi i}\right)^2\left(1+o(1)\right),\qquad\mbox{as $n\to\infty$,}
\end{multline} 
where $G$ is Barnes' $G$-function. This implies that the probability of having all thinned eigenvalues on $\gamma$ decays exponentially as $n\to\infty$, which is  not surprising since there are, before thinning, on average $\left(1-\frac{L}{\pi}\right) n$ CUE eigenvalues outside $\gamma$. The factor $s^{n(1 - \frac{L}{\pi})}$ corresponds to the probability that $\left(1-\frac{L}{\pi}\right) n$ eigenvalues are removed.
If we allow $L$ and $s$ to depend on $n$, several situations are possible. For instance, one observes that the decay in \eqref{as FH} becomes weaker if we let the removal probability $s\to 1$: the decay is destroyed when $1-s=\bigO(1/n)$. Similarly, if we let, for fixed $s\in(0,1)$, the half arclength $L\to \pi$, the decay becomes weaker.
As $s\to 0$, the decay becomes faster instead. In Section \ref{section: statement of results}, we will give a detailed overview of known asymptotic results for the Toeplitz determinants $D_n(s,L)$, depending on how $L$ and $s$ vary with $n$, and we will state a new result on the asymptotic behaviour of $D_n(s,L)$ if $s\to 0$ at an appropriate speed as $n\to\infty$.

\paragraph{Conditioning on a gap of the thinned spectrum.}

Let us now adopt an opposite point of view.
Instead of studying statistics of the thinned spectrum, we assume that we have information about the incomplete thinned CUE spectrum, and we want to deduce statistical information about the complete non-thinned spectrum. In this context, it is natural to interpret the thinned eigenvalues as {\em observed particles}, and the others as {\em unobserved particles}. The general question which we want to investigate, is: {\em what do the observed particles tell us about the unobserved ones?}

We consider the situation where we observe no particles outside a measurable subset $\gamma$ of the unit circle, equivalently $\Phi\subset\gamma$. We now define a conditional CUE as the joint probability distribution of all $n$ eigenvalues in the complete spectrum $\Theta$, given the fact that $\Phi\subset\gamma$. In this conditional CUE, any eigenvalue configuration $\Theta$ can occur, but configurations with many eigenvalues on $\gamma^c$ are less likely because all of them have to be removed by the thinning procedure. We note that, for finite $n$, the event $\Phi\subset\gamma$ has non-zero probability, so the conditional probability is well-defined in the classical sense.

From the definition of conditional probability $\mathbb P(A|B)=\mathbb P(A\cap B)/\mathbb P(B)$, it is straightforward to see that the conditional joint probability distribution for the eigenvalues is given by
\begin{equation}\label{jpdf conditional CUE}
\frac{1}{Z_n}\prod_{1\leq j<k\leq n}|e^{i\theta_j}-e^{i\theta_k}|^2 \prod_{j=1}^nf(e^{i \theta_j})d\theta_j,\qquad \theta_1,\ldots, \theta_n \in [0,2\pi),
\end{equation}
where $Z_n=Z_n(s,\gamma)$ is a normalization constant, and $f$ is given by \eqref{def symbol}. We will refer to this probability measure as the {\em conditional CUE}.
This is again a determinantal point process; the conditional eigenvalue correlation kernel can be expressed as
\begin{equation}\label{kernel conditional CUE}
K_n(e^{i\theta},e^{i\mu})=\frac{1}{2\pi}\sum_{k=0}^{n-1} \frac{1}{h_{k}} \phi_k(e^{i\theta}) \overline{\phi_k(e^{i\mu})} \sqrt{f(e^{i\theta}){f(e^{i\mu})}},
\end{equation}
where $\phi_k, k=0,1,\ldots$ are the monic orthogonal polynomials on the unit circle with respect to the weight $f$, characterized by the orthogonality conditions
\begin{equation}\label{orthogonality}
\frac{1}{2\pi}\int_0^{2\pi} \phi_k(e^{i\theta})\overline{\phi_\ell(e^{i\theta})}f(e^{i\theta})d\theta=h_k\delta_{k\ell}.
\end{equation}
The polynomial $\phi_n$ is moreover the average characteristic polynomial of the conditional CUE. By the Christoffel-Darboux formula for orthogonal polynomials on the unit circle, see e.g.\ \cite{Simon}, we can express $K_n$ in terms of $\phi_n$ only as follows:
\begin{equation}\label{kernel conditional CUE CD}
K_n(e^{i\theta},e^{i\mu})=\frac{1}{2\pi h_{n}} \sqrt{f(e^{i\theta}){f(e^{i\mu})}}\frac{\phi_n^*(e^{i\theta})\overline{\phi_n^*(e^{i\mu})}-\phi_n(e^{i\theta})\overline{\phi_n(e^{i\mu})}
}{1-e^{i(\theta-\mu)}},
\end{equation}
where $\phi_n^*(z)=z^n\overline{\phi_n}(z^{-1})$ is the reverse polynomial. If $f$ has the symmetry $f(e^{-i\theta}) = f(e^{i\theta})$ then $\overline{\phi_{n}}(z) = \phi_{n}(z)$.

Conditional eigenvalue correlations and gap probabilities can be computed from the kernel $K_n$. The one-point function or mean eigenvalue density $\psi_{n,s,L}$ is given by
\begin{align}\label{onepointfunction}
\psi_{n,s,L}(e^{i\theta}) &= \frac{1}{n} K_n(e^{i\theta},e^{i\theta}) \nonumber \\&= \frac{1}{2\pi n h_{n}}f(e^{i\theta})e^{i\theta} \left( \phi_n'(e^{i\theta}) \overline{\phi_n(e^{i\theta})} - \left( \phi_n^* \right)' (e^{i\theta}) \overline{\phi_n^*(e^{i\theta})} \right).
\end{align} 
Formulas \eqref{kernel conditional CUE CD} and \eqref{onepointfunction} are particularly convenient for asymptotic analysis as $n\to\infty$, since they require only to know the large $n$ behaviour of the average characteristic polynomial $\phi_n$. Other interesting quantities are the conditional average and variance of the number of (unobserved) eigenvalues which lie on $\gamma^c$, given that no thinned (observed) eigenvalues lie on $\gamma^c$.
It is clear that, for $s=0$, this number
 is equal to $0$ with probability $1$ since all eigenvalues are observed. If $s=1$, no particles are observed, and then one shows easily that the average is $\left(1-\frac{L}{\pi}\right)n$ by rotational symmetry.

For a general removal probability $s\in(0,1)$, we have
\begin{align*}
\mathbb E\left(\#(\Theta\cap \gamma^c)|\Phi\subset\gamma\right)&=\sum_{k=0}^n k \mathbb P\left(\#(\Theta\cap \gamma^c)=k|\Phi\subset\gamma\right)\\&=\sum_{k=0}^n k \frac{\mathbb P\left(\Phi\subset\gamma \mbox{ and } \#(\Theta\cap \gamma^c)=k\right)}{\mathbb P\left(\Phi\subset\gamma\right)}.
\end{align*}
Using \eqref{gap probability} and \eqref{gap prob Toeplitz}, we obtain the identity
\begin{equation}\label{identity average Toeplitz}
\mathbb E\left(\#(\Theta\cap \gamma^c)|\Phi\subset\gamma\right)=\sum_{k=0}^n ks^k \frac{\mathbb P\left(\#(\Theta\cap \gamma^c)=k\right)}{\mathbb P\left(\Phi\subset\gamma\right)}=s\frac{\partial}{\partial s}\log D_n(f).
\end{equation}
Similarly, after a straightforward calculation,
\begin{align}
\Var\left(\#(\Theta\cap \gamma^c)|\Phi\subset\gamma\right)&=\mathbb E\left([\#(\Theta\cap \gamma^c)]^2|\Phi\subset\gamma\right)-\mathbb E\left(\#(\Theta\cap \gamma^c)|\Phi\subset\gamma\right)^2\nonumber\\
&=\sum_{k=0}^n k^2s^k \frac{\mathbb P\left(\#(\Theta\cap \gamma^c)=k\right)}{\mathbb P\left(\Phi\subset\gamma\right)}
-\left(s\frac{\partial}{\partial s}\log D_n(f)\right)^2\nonumber\\
&=s^2\frac{\partial^2}{\partial s^2}\log D_n(f)+ s\frac{\partial}{\partial s}\log D_n(f).\label{identity variance}
\end{align}
In the case where $s\in(0,1)$ is independent of $n$ and where $\gamma$ is an arc of arclength $2L$ independent of $n$, we obtain from \eqref{as FH}  that
\begin{align}&\label{average FH}
\mathbb E\left(\#(\Theta\cap \gamma^c)|\Phi\subset\gamma\right)=n \left( 1 - \frac{L}{\pi} \right)+\frac{\log s}{\pi^2}\log n+\bigO(1),\\ \label{variance FH}
&\Var\left(\#(\Theta\cap \gamma^c)|\Phi\subset\gamma\right)=\frac{\log n}{\pi^2}+\bigO(1),
\end{align}
as $n\to\infty$, where the $\bigO(1)$ term can be expressed in terms of Barnes' $G$ function. Remarkably, the leading order terms of the average and number variance are independent of $s$, which suggests that the conditioning does not have a major effect on the CUE eigenvalues, in other words the fact that no particles are observed outside $\gamma$ does not give us much additional information on all particles. As we will see below, this changes drastically if we allow $L$ and $s$ to depend on $n$.

\begin{remark}\label{remark: diff}
It is a priori not rigorously justified to take logarithmic derivatives of the asymptotic expansion \eqref{as FH} in order to obtain \eqref{average FH} and \eqref{variance FH}. However, this can be justified in two ways. First, the asymptotics for the Toeplitz determinants $D_n(s,L)$ in \cite{DIK2} were obtained precisely by integrating asymptotics for $\frac{\partial}{\partial s}\log D_n(s,L)$, see e.g.\ \cite[Section 3]{DIK2}. Secondly, the asymptotics \eqref{average FH} are valid for $s$ in a complex neighbourhood of $(0,1)$ with a $o(1)$ error term which is analytic in $s$ and uniform in compact subsets of this neighbourhood. From Cauchy's integral formula, it then follows that $s$-derivatives of the $o(1)$ term are also $o(1)$. This allows one to justify \eqref{average FH} and \eqref{variance FH}.
\end{remark}

We will study the large $n$ asymptotics for various quantities in the conditional CUE:
\begin{enumerate}
\item the zero counting measure $\nu_{n,s,L}$ of the conditional average characteristic polynomial $\phi_n$,
\item the conditional mean eigenvalue density
\[
\frac{d\mu_{n,s,L}(e^{i\theta})}{d\theta}=\psi_{n,s,L}(e^{i\theta})=\frac{1}{n}K_n(e^{i\theta},e^{i\theta}),
\]
\item the conditional eigenvalue correlation kernels, suitably scaled near different types of points on the unit circle,
\item the conditional average and variance of the number $\#(\Theta\cap \gamma^c)$, i.e.\ the number of particles on $\gamma^c$.
\end{enumerate}

Similarly to the gap probabilities of the thinned spectrum, the above quantities exhibit different types of asymptotic behaviour, depending on how $L$ and $s$ behave as $n\to\infty$. We will give an overview of known results in Section \ref{section: statement of results}, and we will in addition derive new results if $s\to 0$ and $n\to \infty$ simultaneously at an appropriate speed.

\section{Statement of results}\label{section: statement of results}

In what follows, we take $\gamma$ to be an arc of the unit circle with arclength $2L\in(0,2\pi)$. By rotational invariance, we can assume without loss of generality that $\gamma$ is an arc between $z_{0} = e^{ iL}$ and $\overline{z_{0}} = e^{ -iL}$, such that $1\in\gamma$. Note that in this case $\overline{\phi_{n}} = \phi_{n}$.

\subsection{Large $n$ asymptotics for gap probabilities in the thinned CUE}

We describe the asymptotics of the quantity $\mathbb P(\Phi\subset\gamma)=D_n(s,L)$ in $4$ different regimes depending on the behaviour of $s$ and $L$. Recall that
\begin{equation}\label{def Toeplitz}
D_{n}(s,L) = \det(f_{k-j})_{j,k=1,...,n},
\end{equation}
where the $f_k$'s are the Fourier coefficients of $f$ given in \eqref{def symbol}. They can be computed explicitly, we have
\begin{equation}\label{def fk}
f_{k} = \frac{1}{2\pi} \int_{0}^{2\pi} f(e^{i\theta})e^{-ik\theta}d\theta = \left\{ 
\begin{array}{l l}
\frac{L}{\pi} + \frac{\pi - L}{\pi}s, & \mbox{ if } k = 0, \\[2mm]
\frac{(1-s)\sin(kL)}{\pi k}, & \mbox{ if } k \neq 0. \\
\end{array} \right.
\end{equation}

\paragraph{Case I: $s$ and $L$ fixed.} The case where $L\in(0,\pi)$ and $s\in(0,1)$ are independent of $n$ was already discussed in the introduction: we have \eqref{as FH}, which implies that $\mathbb P(\Phi\subset\gamma)$ decays exponentially fast as $n\to\infty$. The precise rate of decay depends on the values of $s$ and $L$. The decay is slower for larger $s$ and larger $L$.

\paragraph{Case II: $s$ fixed, $L\to \pi$.} If $s\in(0,1)$ is independent of $n$ and $L\to \pi$, we can use large $n$ asymptotics for $D_n(s,L)$ obtained in \cite[Theorem 1.5]{ClKr} (our situation corresponds to $\alpha_1=\alpha_2=0$, $\beta_1=-\beta_2=\frac{\log s}{2\pi i} $, and $t=\pi -L$ in the language of \cite{ClKr}). This leads to
\begin{equation}\label{Toeplitz PV}
P\left(\Phi\subset\gamma\right) = \exp\left( \int_{0}^{-2in(\pi-L)} \frac{\sigma(\xi;s)}{\xi}d\xi \right) (1+o(1)),
\end{equation} 
as $n\to\infty$ and simultaneously $L\to \pi_-$. Here, $\sigma(\xi;s)$ is a solution of the $\sigma$-form of the Painlev\'e V equation 
\begin{equation}\label{PV}
\xi^{2} (\sigma'')^{2} = (\sigma - \xi \sigma' + 2 (\sigma')^{2})^{2} - 4 (\sigma')^{4},
\end{equation}
where $'$ denotes a derivative with respect to $\xi$, and we have the asymptotic behaviour (depending on $s$)
\begin{equation}\label{PV as}
\begin{array}{r c l l}
\displaystyle \sigma(\xi;s) & = & \displaystyle \bigO(\xi \log \xi ), & \displaystyle \xi \to i0_{-}, \\
\displaystyle \sigma(\xi;s) & = & \displaystyle -\frac{\log s}{2\pi i}\xi + \frac{(\log s)^{2}}{2\pi^{2}} +\bigO(\xi^{-1}), & \displaystyle \xi \to -i\infty.
\end{array}
\end{equation}
In the special case $L=\pi\left(1-\frac{4 y}{n}\right)$ with $y>0$, this can also be derived from the classical results by Jimbo, Miwa, M\^ori, and Sato on the sine kernel Fredholm determinant \cite{JMMS}, recall the discussion above \eqref{partition}. With this scaling of $L$, $\mathbb P(\Phi\subset\gamma)$ does not decay as $n\to\infty$ but tends to a constant depending on the values of $y$ and $s$.
Asymptotics of the integral at the right hand side of \eqref{Toeplitz PV} as $n(\pi-L)\to\infty$ are described in \cite[Formula (1.26)]{ClKr}.

\paragraph{Case III: $s\to 0$ sufficiently fast and $L$ fixed or $L\to\pi$ sufficiently slow.} If $L\in (0,\pi)$ is fixed and $s\to 0$ sufficiently fast such that \begin{equation}\label{def xc}\frac{1}{n}\log s\leq -x_c:=2\log\tan\frac{L}{4},
\end{equation} we can use the results from \cite[Theorem 1.1]{ChCl} to obtain
\begin{equation}\label{as Widom}
\mathbb P\left(\Phi\subset\gamma\right)=\left(\sin\frac{L}{2}\right)^{n^2}n^{-1/4}\left(\cos\frac{L}{2}\right)^{-1/4}2^{\frac{1}{12}}e^{3\zeta'(-1)}\left(1+o(1)\right),
\end{equation}
as $n\to\infty$, where $\zeta$ is the Riemann $\zeta$-function. This implies that the gap probability decays faster than exponentially in this case. The estimate \eqref{as Widom} is the same as in the case $s=0$, in which the asymptotics for the Toeplitz determinants were obtained in \cite{Widom}. This estimate remains valid if the arclength $L$ tends to $\pi$ as $n\to\infty$, as long as $n(\pi-L)\to \infty$ \cite{ChCl}.

In this case, a very small fraction of eigenvalues are removed for $n$ large. In fact, the expected number of removed eigenvalues tends to $0$ exponentially fast as $n\to\infty$, so typically we observe all eigenvalues, and it is not surprising that the gap probability behaves like in the case $s=0$. However, this is no longer true if $s\to 0$ at a slower rate.

\paragraph{Case IV: $s\to 0$ at a slower rate, $L$ fixed.} 
In the case where $s\to 0$ in such a way that $\log s=-xn$ with $0<x<x_c$, asymptotics for $D_n(s,L)$ are not available in the literature. We will study this case in detail and prove the following result.
\begin{theorem}\label{theorem: Toeplitz}
Let $D_n(s,L)$ be the Toeplitz determinant defined in \eqref{def Toeplitz}-\eqref{def fk}.
If $0 \leq x \leq x_c$ with $x_c$ as in \eqref{def xc}, we have
\begin{equation} \label{eq_Toeplitz}
\lim_{n\to\infty}\frac{1}{n^2}\log D_{n}(e^{-xn},L) = - \int_{0}^{x} \Omega_{\xi,L}d\xi.
\end{equation}
Here $\Omega_{\xi,L}$ is given by 
\begin{equation}\label{def Omega}
\Omega_{\xi,L}=\frac{1}{2\pi}\int_{\pi-T}^{\pi+T}\sqrt{\frac{\cos\theta+\cos T}{\cos \theta -\cos L}}d\theta,
\end{equation}
where $0\leq T\leq \pi-L$ is the unique solution of the equation
\begin{equation} \label{x theta_1 equation}
\frac{1}{\pi}\int_{[-L,L]\cup[\pi -T,\pi+T]}  \log \left| \frac{1+e^{i\theta}}{1-e^{i\theta}} \right| \sqrt{\frac{\cos\theta+\cos T}{\cos \theta -\cos L}}d\theta = \xi.
\end{equation}\end{theorem}
\begin{remark}
The fact that equation \eqref{x theta_1 equation} has a unique solution was shown in \cite[Proposition 3.1]{ChCl}. 
\end{remark}
\begin{remark}
The decay of $\mathbb P\left(\Phi\subset\gamma\right)=D_n(e^{-xn},L)$ as $n\to\infty$ is superexponential. Setting $x=x_c$ in \eqref{eq_Toeplitz}, we recover 
\[P\left(\Phi\subset\gamma\right)=D_n(e^{-xn},L)=\left(\sin\frac{L}{2}\right)^{n^2}o(\exp(n^2))\] as in \eqref{as Widom}, provided that the identity
\begin{equation}
\label{identity intOmega}
\int_0^{x_c}\Omega_{\xi,L}d\xi=-\log\sin\frac{L}{2}
\end{equation}
holds. We will give an independent proof of this identity in Appendix \ref{appendix: proofOmega}.
As $x\to 0$, by \eqref{eq_Toeplitz}, the superexponential decay disappears, which makes the connection with the behaviour in case I. Further terms in the asymptotic expansion of $D_n(e^{-xn},L)$ are expected to contain quantities related to elliptic $\theta$-functions, but are hard to compute explicitly.
\end{remark}

\begin{remark}
The quantities $\Omega_{x,L}$ and $T$ are closely related to an equilibrium problem.
Define $\mu_{x,L}$ as the measure
supported on the two disjoint arcs $\theta\in [-L,L]\cup[\pi -T,\pi+T]$ with density
\begin{equation}\label{def muxL}
\frac{d\mu_{x,L}(e^{i\theta})}{d\theta}=
\frac{1}{2\pi}\sqrt{\frac{\cos\theta+\cos T}{\cos \theta -\cos L}}=:\psi_{x,L}(e^{i\theta}),
\end{equation}
see Figure \ref{figure:eqdens}.
It was shown in \cite[Proposition 3.1]{ChCl} that $\mu_{x,L}$ 
is the unique equilibrium measure minimizing the logarithmic energy
\begin{equation}\label{energy}
\iint \log |z-u|^{-1} d\mu(z)d\mu(u) + \int V(z)d\mu(z),
\end{equation}
among all Borel probability measures $\mu$ on the unit circle, where 
\begin{equation}\label{V}
V(e^{i\theta}) = \begin{cases}
0, & \mbox{ for } e^{i\theta} \in \gamma, \\
x, & \mbox{ for } e^{i\theta} \in S^{1} \setminus \gamma.
\end{cases}
\end{equation}
It is characterized by the 
Euler-Lagrange variational conditions
\begin{equation}\label{EulerLagrange}
\begin{array}{r r}
\displaystyle 2 \int_{-\pi}^{\pi} \log |z-e^{i\theta}|d\mu_{x,L}(e^{i\theta}) - V(z) + \ell_{x,L} = 0, & \displaystyle \mbox{ for } z \in \supp\,\mu_{x,L}, \\
\displaystyle 2 \int_{-\pi}^{\pi} \log |z-e^{i\theta}|d\mu_{x,L}(e^{i\theta}) - V(z) + \ell_{x,L} < 0, & \displaystyle \mbox{ for } z \in S^{1}\setminus \supp\,\mu_{x,L}, \\
\end{array}
\end{equation} where  \begin{equation}\label{def ell}
\ell_{x,L} = - \int_{0}^{1} \frac{1}{u} \left( 1 - \sqrt{\frac{u^2 + 2u\cos T  +1}{u^2 - 2u\cos L +1}} \right) du. 
\end{equation}
We have $\Omega_{x,L}=\int_{\pi-T}^{\pi+T} d\mu_{x,L}(e^{i\theta})$.
For $x\geq x_c$, the measure $\mu_{x,L}=\mu_{\infty,L}$ minimizing \eqref{energy} does not depend on the value of $x$ and is supported on the single arc $\theta\in[-L,L]$.
\end{remark}
\begin{figure}[t] \includegraphics[width=1\textwidth]{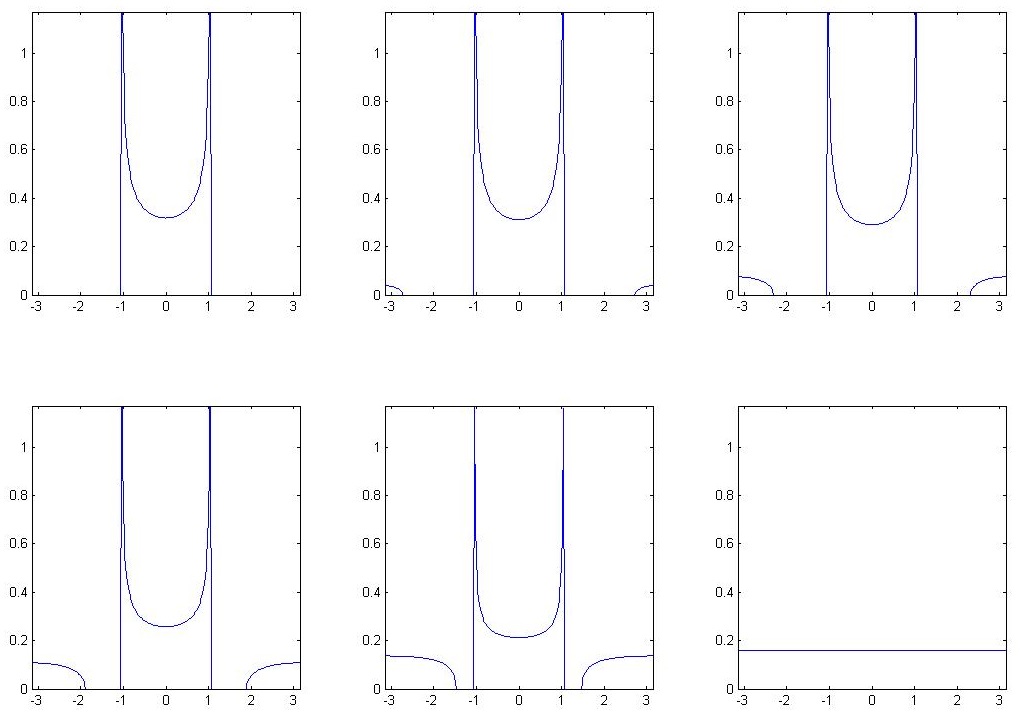} 
\caption{\label{figure:eqdens}Density of the equilibrium measure $\mu_{x,L}$ as a function of $e^{i\theta}$ with $\theta\in[-\pi,\pi]$, for $L=\pi/3$ and for different values of $x$: the upper left picture corresponds to $x \geq x_c$, then $x$ decreases from left to right and from top to bottom, and the bottom right picture corresponds to $x=0$.}
\end{figure}

\subsection{Large $n$ asymptotics in the conditional CUE}

We now take a look at the large $n$ behaviour of the zero counting measure of the average characteristic polynomial $\phi_n$, the limiting mean eigenvalue distribution,  scaling limits of the eigenvalue correlation kernels, and the average and variance of the number of particles on $\gamma^c$, in the conditional CUE.
All these objects can be expressed in terms of the average characteristic polynomial $\phi_n$ and the Toeplitz determinant $D_n(s,L)$, as explained in Section \ref{sec: intro} (recall formulas \eqref{kernel conditional CUE CD}--\eqref{identity variance}).
We denote \begin{equation}\label{def mun}
\mu_{n,s,L}(z) := \mathbb{E} \left( \frac{1}{n} \sum_{j=1}^n \delta (z - e^{i\theta_j}) \right)
\end{equation}
for the average counting measure of the eigenvalues in the conditional CUE, with density $\psi_{n,s,L}$ given by \eqref{onepointfunction}, and we let $\nu_{n,s,L}$ be the zero counting measure of the average characteristic polynomial $\phi_n$.

We again distinguish between the $4$ different regimes identified before. In Cases I-III, large $n$ asymptotics can be derived from results for the Toeplitz determinants and average characteristic polynomials that are (explicitly or implicitly) available in the literature. In Case IV, large $n$ asymptotics for $\phi_n$ and $D_n$ are not available in the literature. We will perform a detailed asymptotic analysis using Riemann-Hilbert (RH) methods in Section \ref{section: RH}, and this will enable us to obtain asymptotic results in the conditional CUE.

\paragraph{Case I: $s$ and $L$ fixed.} 

An asymptotic analysis of $\phi_n$ in this case was done in \cite{DIK}. Using the results obtained in that paper, we will show in Section \ref{section: proofs} that $\nu_{n,s,L}$ and $\mu_{n,s,L}$, for $0<s<1$, converge weakly to the uniform measure on the unit circle; on the level of densities, we have
\begin{equation}\label{psi s fixed} \lim_{n\to\infty} \psi_{n,s,L} (e^{i\theta}) = \frac{1}{2\pi},\end{equation}
uniformly for $\theta$ in compact subsets of $[-\pi,\pi]\setminus \{ -L,L \}$.
Note that for $s=1$, we have $\phi_n(z)=z^n$ and then $\nu_{n,1,L}$ does not converge to the uniform measure on the circle, since it is a Dirac measure at $0$.

For the eigenvalue correlation kernel $K_n$, one expects from the analysis in \cite{DIK} to have the scaling limit
\begin{equation}\label{sine kernel limit}
\lim_{n\to\infty}\frac{1}{cn}K_n\left(e^{i\left(\theta+\frac{u}{cn}\right)}, e^{i\left(\theta+\frac{v}{cn}\right)}\right) = e^{i \frac{u-v}{2c}} \frac{\sin \pi(u-v)}{\pi(u-v)},
\end{equation}
for any $u,v\in\mathbb R$ and $\theta\in [0,2\pi]\setminus\{L,2\pi-L\}$, and where $c=\frac{1}{2\pi}$. Near $e^{\pm iL}$, the asymptotics for $\phi_n$ are described in terms of confluent hypergeometric functions, and instead of the sine kernel, one expects a limiting kernel built out of confluent hypergeometric functions.
These scaling limits can be derived from the asymptotics for $\phi_n$, using similar computations as the ones in Sections \ref{section: proof Bessel}--\ref{section: sine} (where we prove scaling limits in Case IV). The asymptotics for the expected number of eigenvalues on $\gamma^c$ were already given in \eqref{average FH}, and for the number variance in \eqref{variance FH}.

The above results can heuristically be explained as follows: the expected fraction of CUE eigenvalues on $\gamma^c$ is equal to $1-\frac{L}{\pi}$.
From the theory of large deviations, see \cite{AGZ}, it follows that the probability of having a smaller fraction of eigenvalues on $\gamma^c$ decays exponentially fast in $n^2$ as $n\to\infty$. On the other hand, the probability that all eigenvalues on $\gamma^c$ are removed by the thinning procedure is only exponentially small in $n$. In other words, if there are no observed particles on $\gamma^c$, that does not say much about the complete spectrum, it just means that the thinning procedure has accidentally removed all the eigenvalues on $\gamma^c$.
It is expected that the eigenvalues uniformly fill out the arc $\gamma^c$ as $n\to \infty$, just like in the CUE without conditioning. The conditioning only has an effect on the local behaviour of eigenvalues near $e^{\pm iL}$: near those points, eigenvalues are more likely to lie on $\gamma$ than outside.

\paragraph{Case II: $s$ fixed, $L\to \pi$.} 

Here, asymptotics for $\phi_n$ were obtained in \cite{ClKr}. We will show in Section \ref{section: proofs}, relying on the analysis in \cite{ClKr}, that  $\mu_{n,s,L}$ and $\nu_{n,s,L}$ converge weakly to the uniform measure on the unit circle. For the correlation kernel, we expect convergence to the sine kernel as in \eqref{sine kernel limit}, except when $\theta=\pi$, in which case one expects a limiting kernel built out of functions related to the Painlev\'e V equation, see  \cite{ClaeysFahs} for a similar situation on the real line.

Taking logarithmic derivatives of \eqref{Toeplitz PV} with respect to $s$, we obtain by \eqref{identity average Toeplitz}--\eqref{identity variance} that the average number of particles on $\gamma^c$ and the number variance are given by
\begin{align}
&\mathbb E\left(\#(\Theta\cap \gamma^c)|\Phi\subset\gamma\right)=s\int_{0}^{-2in(\pi-L)} \frac{\partial_s\sigma(\xi;s)}{\xi}d\xi  +o(1),\\
&\Var\left(\#(\Theta\cap \gamma^c)|\Phi\subset\gamma\right)=\int_{0}^{-2in(\pi-L)} \frac{s^2\partial_s^2\sigma(\xi;s)+s\partial_s\sigma(\xi;s)}{\xi}d\xi  +o(1),
\end{align}
as $n\to\infty$, where $\sigma(\xi;s)$ is the Painlev\'e V transcendent defined in \eqref{PV}--\eqref{PV as}.
This can be justified rigorously as described in Remark \ref{remark: diff}. 

The case where $L=\pi\left(1-\frac{4y}{n}\right)$ with $y>0$ independent of $n$ is of particular interest because the probability of having no eigenvalues on $\gamma^c$ does not decay as $n\to\infty$, and this means that we condition on an event which is not unlikely to occur for large $n$. The average and variance of the number of particles on $\gamma^c$ are described in terms of a Painlev\'e V transcendent which depends on the removal probability $s$. Here, the fact that we do not observe particles on $\gamma^c$ does have implications for the unobserved particles.

\paragraph{Case III: $s\to 0$ sufficiently fast and $L$ fixed or $L\to\pi$ sufficiently slow.}

As $n\to\infty$ and at the same time $s$ goes to $0$ in such a way that $s\leq \left(\tan\frac{L}{4}\right)^{2n}$, asymptotics for $\phi_n$ were obtained in \cite{ChCl}. Both $\mu_{n,s,L}$ and $\nu_{n,s,L}$ converge weakly to the measure
\begin{equation}\label{eq measure}
d\mu_{\infty,L}(e^{i\theta}) := \frac{1}{2\pi} \sqrt{\frac{\cos\theta+1}{\cos\theta -\cos L}} d\theta, \qquad \theta \in [-L,L],
\end{equation}
see Section \ref{section: proofs}.
Moreover,
\begin{equation}\label{lim density}
\lim_{n\to\infty}\psi_{n,s,L}(e^{i\theta})  = \frac{1}{2\pi} \sqrt{ \frac{ \cos \theta + 1}{\cos \theta - \cos L}} \chi_{[-L,L]}(\theta),
\end{equation}
uniformly for $\theta$ in compact subsets of $[-\pi,\pi]\setminus \{ -L,L \}$. 
For the eigenvalue correlation kernel, we expect the sine kernel scaling limit \eqref{sine kernel limit} with $c=\frac{1}{2\pi} \sqrt{ \frac{ \cos \theta + 1}{\cos \theta - \cos L}}$ for $\theta\in (-L,L)$.
Near the edges of $\gamma$, we obtain the Bessel kernel,
\begin{multline}\label{Bessel kernel limit}
\lim_{n\to\infty} \frac{1}{(cn)^2} K_n \left( e^{\pm i\left(L-\frac{u}{(cn)^2}\right)}, e^{\pm i\left(L-\frac{v}{(cn)^2}\right)}\right)\\ = \frac{J_{0}(\sqrt{u})\sqrt{v}J_{0}^{\prime}(\sqrt{v})-J_{0}(\sqrt{v})\sqrt{u}J_{0}^{\prime}(\sqrt{u})}{2(u-v)},
\end{multline}
for $u,v>0$, where $c = \sqrt{\cot \frac{L}{2}}$.
This is remarkable because the Bessel kernel usually appears near hard edges of the spectrum. In our situation, for finite $n$ and $s\in(0,1)$, it can happen that there are eigenvalues on $\gamma^c$, so we do not have a hard edge.
Taking the logarithmic $s$-derivative of \eqref{as Widom}, we obtain by \eqref{identity average Toeplitz}--\eqref{identity variance},
\begin{align}&\label{average Widom}
\mathbb E\left(\#(\Theta\cap \gamma^c)|\Phi\subset\gamma\right)=o(s),\\
&\Var\left(\#(\Theta\cap \gamma^c)|\Phi\subset\gamma\right)=o(s)
,
\end{align}
as $n\to\infty$.
In other words, the expected number of eigenvalues on $\gamma^c$ is exponentially small, as well as the number variance.

In this case, observing no particles on $\gamma^c$ implies that it is unlikely to have unobserved particles there too. This is natural because typically no or very few particles are removed by the thinning procedure.
Then, because of the repulsion between the eigenvalues, there tend to be more eigenvalues near the edges $e^{\pm i L}$ of $\gamma$, which heuristically explains the blow-up of the limiting mean eigenvalue density.
Asymptotically for large $n$, the eigenvalues in the conditional CUE behave locally near $e^{\pm  iL}$ like near a hard edge.

\paragraph{Case IV: $s\to 0$ at a slower rate, $L$ fixed.} 

This is the intermediate regime between Case I and Case III. In contrast to Case I, not observing particles outside $\gamma$ does have an effect on the large $n$ macroscopic behaviour of the eigenvalues, but it is not as drastic as in Case III. Typically, for large $n$, there will be a non-zero fraction of eigenvalues outside $\gamma$, but this fraction is smaller than in Case I.
We prove the following results.

\begin{theorem}\label{theorem macro}
Fix $0<L<\pi$.
We take $s=e^{-xn}$ with $0<x < x_c = - 2 \log \tan \frac{L}{4}$.
\begin{enumerate}
\item The measures $\mu_{n,s,L}$ and $\nu_{n,s,L}$ converge weakly to the measure $\mu_{x,L}$ given by \eqref{def muxL} as $n\to\infty$; moreover,
\begin{equation}
\lim_{n\to\infty} \psi_{n,s,L}(e^{i\theta}) 
= \frac{d\mu_{x,L}(e^{i\theta})}{d\theta},
\end{equation}
uniformly for $\theta$ in compact subsets of $[-\pi,\pi]\setminus \{ -\pi+T,-L,L,\pi-T \}$.
\item We have
\begin{align}&\label{limOmega}
\lim_{n\to\infty}\frac{1}{n}\mathbb E\left(\#(\Theta\cap \gamma^c)|\Phi\subset\gamma\right)=\int_{\pi-T}^{\pi+T}d\mu_{x,L}(e^{i\theta})d\theta=\Omega_{x,L},\\
&\lim_{n\to\infty}\frac{1}{n}\Var\left(\#(\Theta\cap \gamma^c)|\Phi\subset\gamma\right)=0.\label{thm variance}
\end{align}
\end{enumerate}
\end{theorem}
\begin{remark}
We will also provide a result on the location of the zeros of the orthogonal polynomials: we will show that $\phi_n$ has, for $n$ sufficiently large, all its zeros either close to the unit circle or to a point $w_{n}$ on the real line, see Proposition \ref{prop: zeros phin} below. A similar but stronger result is also shown in Case I and Case III, see Proposition \ref{prop: zeros phin2}.
\end{remark}

\begin{remark}
The eigenvalues outside $\gamma$ are expected on an arc containing $-1$, and it is unlikely to have eigenvalues close to the edges of $\gamma^c$. This is at first sight surprising, but can be explained (a posteriori) by the repulsion between the eigenvalues. As $x\to x_c$, $T\to 0$ and the average of the number of particles on $\gamma^c$ becomes smaller; as $x\to 0$, $\mu_{x,L}$ converges weakly to the uniform measure on the unit circle, and we recover the leading order of the corresponding result in Case I. It is worth noting that the variance is small compared to the average, this means that we can accurately guess the fraction of eigenvalues lying on $\gamma^c$ for large $n$.
\end{remark}

Next, we consider scaling limits of the eigenvalue correlation kernels.

\begin{theorem}\label{theorem kernel}
We let $s=e^{-xn}$ with $x<x_c$.
\begin{enumerate}
\item If $e^{i\theta}$ lies in the interior of the support of $\mu_{x,L}$, we have the sine kernel limit
\eqref{sine kernel limit} with $c=\psi_{x,L}(e^{i\theta})$.
\item At the edges of $\gamma$, we have the Bessel kernel limits
\eqref{Bessel kernel limit}
for $u,v>0$, with $c=\sqrt{\frac{|z_{0}-z_{1}||z_{0}-\overline{z_{1}}|}{|z_{0}-\overline{z_{0}}|}}$, $z_0 = e^{i L}$, and $z_1 = e^{i(\pi-T)}$.
\item At the soft edges, we have the Airy kernel limits
\begin{multline}\label{Airy kernel limit}
\lim_{n\to\infty}\frac{e^{-i \frac{n^{1/3}}{2c^{2/3}}(u-v)}}{(cn)^{2/3}}K_n\left(e^{\pm i\left(- \pi + T+\frac{u}{(cn)^{2/3}}\right)}, e^{\pm i\left(-\pi + T+\frac{v}{(cn)^{2/3}}\right)}\right)\\=\frac{\Ai(u)\Ai^{\prime}(v)-\Ai^{\prime}(u)\Ai(v)}{u-v},
\end{multline}
for $u,v\in\mathbb R$,
with $c=\sqrt{\frac{|\overline{z_{1}}-z_{1}|}{4|\overline{z_{1}}-z_{0}||\overline{z_{1}}-\overline{z_{0}}|}}$. 
\end{enumerate}
\end{theorem}

\begin{remark}
As expected, we obtain the sine kernel in the interior of the support of $\mu_{x,L}$. Near the edges, we have two different types of behaviour: near  $\theta=\pm L$, scaling limits lead to the Bessel kernel, as if there were hard edges; near $\theta=\pi -T$ and $\theta=\pi +T$, we have soft edges with square root vanishing of the limiting mean eigenvalue density, and obtain the usual Airy kernel. 
\end{remark}

\subsubsection*{Outline}

In Section \ref{section: RH}, we will use the Deift/Zhou steepest descent method applied to the RH problem for orthogonal polynomials on the unit circle to obtain large $n$ asymptotics for $\phi_n$ in Case IV. The main ingredients of the RH analysis will be a $g$-function related to the equilibrium measure $\mu_{x,L}$, the construction of standard Airy parametrices near $e^{\pm i(\pi-T)}$, and the construction of Bessel parametrices near $e^{\pm iL}$. Because those points are not hard edges, we will need to use a non-standard Bessel parametrix similar to the one constructed in \cite{ChCl} and in \cite{BDIK}.
The RH analysis yields asymptotics for $\phi_n$, which we describe in detail in Section \ref{section: as phin}.

Next, in Section \ref{section: proofs}, we will use the asymptotics for $\phi_n$ to obtain asymptotics for the zero counting measure $\nu_{n,s,L}$, the mean eigenvalue distribution $\mu_{n,s,L}$, and the eigenvalue correlation kernel.
In this way, we will be able to prove Theorem \ref{theorem: Toeplitz}, Theorem \ref{theorem macro}, and Theorem \ref{theorem kernel}, as well as the corresponding results in Cases I-III.

In Appendix \ref{appendix: proofOmega}, we prove \eqref{identity intOmega}, and in Appendix \ref{appendix: shrinking disks}, we extend the asymptotic analysis from Section \ref{section: RH} to the case of small $x$.

\section{Riemann-Hilbert analysis for $0<x<x_c$}\label{section: RH}

The goal of this section is to obtain large $n$ asymptotics for $\phi_{n}(z)$ for $z$ anywhere in the complex plane.

\subsection{RH problem for orthogonal polynomials on the unit circle}

Consider the matrix $Y$ defined by
\begin{equation}\label{sol_Y}
Y(z) = \begin{pmatrix} \phi_{n}(z) & \displaystyle \int_{S^{1}} \frac{\phi_{n}(w)}{w-z} \frac{f(w)}{2\pi i w^{n}} dw \\
\displaystyle -h_{n-1}^{-1} z^{n-1} \phi_{n-1}(z^{-1}) & \displaystyle -h_{n-1}^{-1} \int_{S^{1}} \frac{\phi_{n-1}(w^{-1})}{w-z} \frac{f(w)}{2\pi i w} dw
\end{pmatrix},
\end{equation}
where $\phi_n$ is the monic orthogonal polynomial with respect to the weight $f$, given by \eqref{def symbol} and depending on $s$ and $L$, and where $h_n$ is the norming constant, defined in \eqref{orthogonality}.

The first column of $Y$ contains the orthogonal polynomials of degree $n$ and $n-1$, while the second column contains their Cauchy transforms with respect to the weight function $f$. $Y$ depends on $n$ and also on $s=e^{-xn}$ and $L$ through the weight function $f$. It is well known that $Y(z)$ is the unique $2 \times 2$ matrix-valued function  which satisfies the following RH problem \cite{FokasItsKitaev}:
\subsubsection*{RH problem for $Y$}
\begin{itemize}
\item[(a)] $Y : \mathbb{C}\setminus S^{1} \to \mathbb{C}^{2\times 2}$ is analytic.
\item[(b)] $Y$ has the following jumps:
\[
Y_{+}(z) = Y_{-}(z) \begin{pmatrix}
1 & z^{-n}f(z) \\ 0 & 1
\end{pmatrix}, \hspace{0.5cm} \mbox{ for } z \in S^{1}\setminus \left\{e^{iL}, e^{-iL}\right\}.
\]
\item[(c)] $Y(z) = \left(I + \bigO(z^{-1})\right) \begin{pmatrix}
z^{n} & 0 \\ 0 & z^{-n}
\end{pmatrix}$ as $z \to \infty$.
\item[(d)] As $z$ tends to $e^{\pm iL}$, the behaviour of $Y$ is
\[
Y(z) = \begin{pmatrix}
\bigO(1) & \bigO(\log |z-e^{\pm i L}|) \\ \bigO(1) & \bigO(\log |z-e^{\pm i L}|)
\end{pmatrix}.
\]
\end{itemize}

We want to obtain large $n$ asymptotics for the matrix $Y$, uniformly for $x$ in compact subsets of $(0,x_{c})$. Such asymptotics can be obtained using the Deift/Zhou steepest descent method \cite{DeiftZhou} applied to the RH problem for $Y$. Following the general procedure developed in \cite{DKMVZ2, DKMVZ1} and applied first to orthogonal polynomials on the unit circle in \cite{BDJ}, we will apply a series of transformations to the RH problem for $Y$, with the goal of obtaining, in the end, a RH problem for which we can easily compute asymptotics of the solution. The first transformations are similar to those in \cite[Section 4]{ChCl}.

\subsection{First transformation $Y \mapsto T$}
We define $g$, depending on $x$ and $L$, by
\begin{equation}\label{g}
g(z)=\int_{[-L,L]\cup [\pi-T,\pi+T]} \log(z-e^{i\theta})d\mu (e^{i\theta}),
\end{equation}
where we write $\mu$ for $\mu_{x,L}$, omitting the subscripts $x, L$, where $\mu_{x,L}$ is the equilibrium measure with density \eqref{def muxL} and satisfying the variational conditions \eqref{EulerLagrange}. In \eqref{g}, for each $\theta$, the branch is chosen such that $\log(z-e^{i\theta})$ is analytic in $\mathbb{C}\setminus ( (-\infty,-1]  \cup \{ e^{it} : -\pi \leq t \leq \theta \} )$ and $\log(x-e^{i\theta}) \sim \log x$ as $x \in \mathbb{R}^{+}$, $x \to \infty$.

The following transformation has the effect of normalizing the RH problem at infinity. We define
\begin{equation}\label{def T}
T(z) = e^{-\frac{n\pi i}{2}\sigma_{3}}e^{\frac{n\ell}{2}\sigma_{3}} Y(z) e^{-ng(z)\sigma_{3}} e^{-\frac{n\ell}{2}\sigma_{3}}e^{\frac{n\pi i}{2}\sigma_{3}},
\end{equation}
with $\ell=\ell_{x,L}$ given by \eqref{def ell}, and $\sigma_3=\begin{pmatrix}1&0\\0&-1
\end{pmatrix}$.
Then $T$ solves the following RH problem:

\subsubsection*{RH problem for $T$}
\begin{itemize}
\item[(a)] $T : \mathbb{C}\setminus S^{1} \to \mathbb{C}^{2\times 2}$ is analytic.
\item[(b)] $T$ satisfies the jump relation
\[
T_{+}(z) = T_{-}(z) J_{T}(z),\qquad \mbox{ on } S^{1}\setminus \left\{e^{iL},e^{-iL}\right\},
\]
with
\[
J_{T}(z) = \begin{pmatrix}
\displaystyle e^{-n(g_{+}(z)-g_{-}(z))} & \displaystyle (-1)^{n}z^{-n}e^{-nV(z)}e^{n\ell}e^{n(g_{+}(z)+g_{-}(z))} \\
\displaystyle 0 & \displaystyle e^{n(g_{+}(z)-g_{-}(z))}
\end{pmatrix}.
\]
\item[(c)] $T(z) = I + O(z^{-1})$ as $z \to \infty$.
\item[(d)] As $z\to e^{\pm i L}$,  we have
\[
T(z) = \begin{pmatrix}
\bigO(1) & \bigO(\log |z-e^{\pm i L}|) \\ \bigO(1) & \bigO(\log |z-e^{\pm i L}|)
\end{pmatrix}.
\]
\end{itemize}

Writing
\begin{align}
&\Sigma_{1} = \{ e^{i\theta} : L < \theta < \pi-T \}, \\ &\Sigma_2 = \{ e^{i\theta} : -\pi + T < \theta < -L \}, \\ &\widetilde{\gamma} =  \{ e^{i\theta} : \pi-T < \theta < \pi + T \},
\end{align}
 we can use \eqref{EulerLagrange} and the definition of $g$ \eqref{g} to write the jump matrices for $T$ in the following form,
\begin{equation}
J_{T}(z) = \left\{ \begin{array}{l l}
\begin{pmatrix}
e^{-2n\pi i \int_{\arg z}^{\pi} d\mu(e^{i\theta})} & 1 \\
0 & e^{2n\pi i \int_{\arg z}^{\pi} d\mu(e^{i\theta})}
\end{pmatrix}, & z \in \gamma\cup\widetilde\gamma, \\
\begin{pmatrix}
e^{-n\pi i \Omega} & e^{n \left[ 2\int_{-\pi}^{\pi} \log |z-e^{i\theta}|d\mu(e^{i\theta}) - x + \ell \right]} \\
0 & e^{n\pi i \Omega}
\end{pmatrix}, & z \in \Sigma_1, \\
\begin{pmatrix}
e^{n\pi i \Omega} & e^{n \left[ 2\int_{-\pi}^{\pi} \log |z-e^{i\theta}|d\mu(e^{i\theta}) - x + \ell \right]} \\
0 & e^{-n\pi i \Omega}
\end{pmatrix}, & z \in \Sigma_2, \\
\end{array} \right. 
\end{equation}
where $\Omega = \int_{\pi-T}^{\pi+T} d\mu(e^{i\theta})$, as in \eqref{def Omega}. We write $z_0=e^{iL}$, $z_1=e^{i(\pi-T)}$ and define
\begin{equation} \label{def_of_phi}
\phi(z) = \int_{z_{0}}^{z} \left(\frac{(\xi-z_{1})(\xi-\overline{z_{1}})}{(\xi-z_{0})(\xi-\overline{z_{0}})}\right)^{1/2} \frac{d\xi}{\xi},
\end{equation}
and
\begin{equation} \label{def_of_tilde_phi}
\widetilde{\phi}(z) = \int_{\overline{z_{1}}}^{z} \left(\frac{(\xi-z_{1})(\xi-\overline{z_{1}})}{(\xi-z_{0})(\xi-\overline{z_{0}})}\right)^{1/2} \frac{d\xi}{\xi},
\end{equation}
where we define the square roots with branch cuts on $\gamma\cup\widetilde\gamma$ and such that they tend to $1$ as $\xi\to\infty$. For $\phi$, we take the path of integration such that it does not cross $\overline{\gamma\cup\widetilde\gamma\cup\Sigma_2}$, and for $\widetilde\phi$, we take it such that it does not cross $\overline{\gamma\cup\widetilde\gamma\cup\Sigma_1}$.
Then, it is straightforward to check that $e^\phi$ is single-valued and analytic on $\mathbb C\setminus\left(\overline{\gamma\cup\widetilde\gamma\cup\Sigma_2} \cup\{0\}\right)$, and that 
$e^{\widetilde\phi}$ is single-valued and analytic on $\mathbb C\setminus\left(\overline{\gamma\cup\widetilde\gamma\cup\Sigma_1} \cup\{0\}\right)$.
We can rewrite the jump matrix $J_{T}$ in terms of $\phi$ and $\widetilde{\phi}$, in the same manner as in \cite[Section 4.1]{ChCl}:
\begin{equation}
J_{T}(z) = \left\{ \begin{array}{l l}
\begin{pmatrix}
\displaystyle e^{n(\phi_{-}(z)-\pi i \Omega)} & \displaystyle 1 \\
\displaystyle 0 & \displaystyle e^{-n(\phi_{-}(z) - \pi i \Omega)}
\end{pmatrix}, & z \in \gamma, \\[0.6cm]
\begin{pmatrix}
\displaystyle e^{n(\widetilde{\phi}_{-}(z)+\pi i \Omega)} & \displaystyle 1 \\
\displaystyle 0 & \displaystyle e^{-n(\widetilde{\phi}_{-}(z) + \pi i \Omega)}
\end{pmatrix}, & z \in \widetilde{\gamma}, \\[0.6cm]
\begin{pmatrix}
\displaystyle e^{-n\pi i \Omega} & \displaystyle e^{n(\phi(z)-x)} \\
\displaystyle 0 & \displaystyle e^{n\pi i \Omega}
\end{pmatrix}, & z \in \Sigma_1, \\[0.6cm]
\begin{pmatrix}
\displaystyle e^{n\pi i \Omega} & \displaystyle e^{n\widetilde{\phi}(z)} \\
\displaystyle 0 & \displaystyle e^{-n\pi i \Omega}
\end{pmatrix}, & z \in \Sigma_2. \\
\end{array} \right.
\end{equation}
Now we are able to define the next transformation: the opening of the lenses.

\subsection{Second transformation $T \mapsto S$}

We can factorize $J_{T}$ on $\gamma$ as follows:
\begin{multline*}
\begin{pmatrix}
e^{n(\phi_{-}(z)-\pi i \Omega)} & 1 \\
0 & e^{-n(\phi_{-}(z)-\pi i \Omega)}
\end{pmatrix}\\ = \begin{pmatrix}
1 & 0 \\ e^{-n(\phi_{-}(z)-\pi i \Omega)} & 1
\end{pmatrix}\begin{pmatrix}
0 & 1 \\ -1 & 0
\end{pmatrix}\begin{pmatrix}
1 & 0 \\ e^{n(\phi_{-}(z)-\pi i \Omega)} & 1
\end{pmatrix}.
\end{multline*}
There is a similar factorization of $J_{T}$ on $\widetilde{\gamma}$.
Using this factorization, we can split the jump on $\gamma$ into three different jumps on a lens-shaped contour, see Figure \ref{fig_S}.
Denote by $\gamma_{+}$ and $\gamma_{-}$ the lenses around $\gamma$ on the $|z|<1$ side and the $|z|>1$ side respectively and similarly by $\widetilde{\gamma}_{+}$ and $\widetilde{\gamma}_{-}$ the lenses around $\widetilde{\gamma}$, see Figure \ref{fig_S}. Define 
\begin{equation}\label{def S}
S(z) =T(z)\times\ \left\{ \begin{array}{l l}

\begin{pmatrix}
1 & 0 \\
-e^{-n\phi(z)}e^{-n\pi i \Omega} & 1
\end{pmatrix}, &
|z|<1, z \mbox{ inside the lenses around }\gamma, \\[0.6cm]
\begin{pmatrix}
1 & 0 \\
e^{-n\phi(z)}e^{n\pi i \Omega} & 1
\end{pmatrix}, &
|z|>1, z \mbox{ inside the lenses around }\gamma, \\[0.6cm]
\begin{pmatrix}
1 & 0 \\
-e^{-n\widetilde{\phi}(z)}e^{n\pi i \Omega} & 1
\end{pmatrix}, &
|z|<1, z \mbox{ inside the lenses around }\widetilde{\gamma}, \\[0.6cm]
\begin{pmatrix}
1 & 0 \\
e^{-n\widetilde{\phi}(z)}e^{-n\pi i \Omega} & 1
\end{pmatrix}, &
|z|>1, z \mbox{ inside the lenses around }\widetilde{\gamma}, \\[0.6cm]
I,  &   z \mbox{ outside the lenses.}                 \\

\end{array} \right.
\end{equation}
Then $S$ solves the following RH problem.
\subsubsection*{RH problem for $S$}

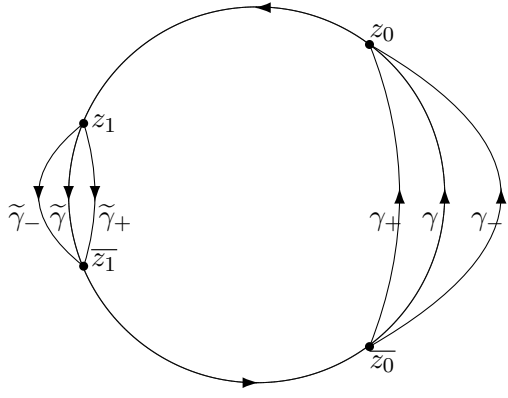
\begin{figure}[t]
    \begin{center}
    \setlength{\unitlength}{1truemm}
    \begin{picture}(100,55)(-5,10)
        \cCircle(50,40){25}[f]
        \put(65,60){\thicklines\circle*{1.2}}
        \put(65,19.8){\thicklines\circle*{1.2}}
	    \qbezier(65,60)(73,40)(65,19.8)
	    \qbezier(65,60)(100,40)(65,19.8)
        \put(65,61){$z_0$}
        \put(72,36){$\gamma$}
        \put(65,36){$\gamma_+$}
        \put(78.5,36){$\gamma_-$}
        \put(22.5,36){$\widetilde{\gamma}$}      
        \put(65,17){$\overline{z_0}$}
        \put(75,41){\thicklines\vector(0,1){.0001}}
        \put(69,41){\thicklines\vector(0,1){.0001}}
        \put(82.5,41){\thicklines\vector(0,1){.0001}}
        \put(25,39){\thicklines\vector(0,-1){.0001}}
        
        \put(27,49.5){\thicklines\circle*{1.2}}
        \put(27,30.5){\thicklines\circle*{1.2}}
        \qbezier(27,49.5)(15,40)(27,30.5)
        \qbezier(27,49.5)(30,40)(27,30.5)
        \put(21,39){\thicklines\vector(0,-1){.0001}}
        \put(28.5,39){\thicklines\vector(0,-1){.0001}}
        \put(17,36){$\widetilde{\gamma}_{-}$}
        \put(29,36){$\widetilde{\gamma}_{+}$}
        \put(28,49){$z_{1}$}
        \put(28,30.5){$\overline{z_{1}}$}
        \put(49.7,64.9){\thicklines\vector(-1,0){.0001}}
        \put(50.3,15){\thicklines\vector(1,0){.0001}}
    \end{picture}
    \caption{The jump contour for $S$.}
    \label{fig_S}
\end{center}

\end{figure}

\begin{itemize}
\item[(a)] $S : \mathbb{C}\setminus (S^{1} \cup \gamma_{+}\cup\gamma_{-}\cup\widetilde{\gamma}_{+}\cup\widetilde{\gamma}_{-}) \to \mathbb{C}^{2\times 2}$ is analytic.
\item[(b)] $S$ satisfies the jump relations
\begin{equation}
\begin{array}{l l}
S_{+}(z) = S_{-}(z) \begin{pmatrix}
0 & 1 \\
-1 & 0
\end{pmatrix}, & \mbox{ for } z \in \gamma\cup\widetilde\gamma, \\

S_{+}(z) = S_{-}(z)  \begin{pmatrix}
\displaystyle e^{-n\pi i \Omega} & \displaystyle e^{n(\phi(z)-x)} \\
\displaystyle 0 & \displaystyle e^{n\pi i \Omega}
\end{pmatrix},  & \mbox{ for } z \in \Sigma_1, \\

S_{+}(z) = S_{-}(z)  \begin{pmatrix}
\displaystyle e^{n\pi i \Omega} & \displaystyle e^{n\widetilde{\phi}(z)} \\
\displaystyle 0 & \displaystyle e^{-n\pi i \Omega}
\end{pmatrix},  & \mbox{ for } z \in \Sigma_2, \\

S_{+}(z) = S_{-}(z) \begin{pmatrix}
1 & 0 \\
e^{-n\phi(z)}e^{-n\pi i \Omega} & 1
\end{pmatrix},  & \mbox{ for } z \in \gamma_{+}, \\

S_{+}(z) = S_{-}(z) \begin{pmatrix}
1 & 0 \\
e^{-n\phi(z)}e^{n\pi i \Omega} & 1
\end{pmatrix},  & \mbox{ for } z \in \gamma_{-}, \\

S_{+}(z) = S_{-}(z) \begin{pmatrix}
1 & 0 \\
e^{-n\widetilde{\phi}(z)}e^{n\pi i \Omega} & 1
\end{pmatrix},  & \mbox{ for } z \in \widetilde{\gamma}_{+}, \\

S_{+}(z) = S_{-}(z) \begin{pmatrix}
1 & 0 \\
e^{-n\widetilde{\phi}(z)}e^{-n\pi i \Omega} & 1
\end{pmatrix},  & \mbox{ for } z \in \widetilde{\gamma}_{-}. \\
\end{array}
\end{equation}
\item[(c)] $S(z) = I + O(z^{-1})$ as $z \to \infty$.
\item[(d)] As $z\to e^{\pm i L}$, we have
\[
S(z)=\bigO(\log|z-e^{\pm i L}|).
\]
\end{itemize}

On $\gamma_{+} \cup \gamma_{-}$ (resp. $\widetilde{\gamma}_{+} \cup \widetilde{\gamma}_{-}$), one shows as in \cite[Section 4.2]{ChCl} that $\Re\phi(z)>0$ (resp. $\Re\widetilde{\phi}(z)>0$) and consequently the jump matrices for $S$ converge to the identity matrix on $\gamma_{+} \cup \gamma_{-}\cup \widetilde{\gamma}_{+} \cup \widetilde{\gamma}_{-}$. 
On $\Sigma_1$, we have that $\Re\left(\phi(z)-x\right)<0$, so that the jump matrix converges to a diagonal matrix. On $\Sigma_2$ finally, we similarly have $\Re\widetilde\phi(z)<0$, and the jump matrix converges to a diagonal matrix here as well.
The convergence of the jump matrices is point-wise in $z$ and breaks down as $z$ approaches $e^{\pm i L}$ and $e^{ i(\pi\pm T)}$.

Therefore, we will construct approximations to $S$ for large $n$ in different regions of the complex plane: local parametrices will be constructed in small disks $D(e^{\pm i L},r)$ and  $D(e^{i(\pi\pm T)},r)$ surrounding the edge points of the support of $\mu$, and a global parametrix elsewhere.

\subsection{Global parametrix}

We will construct the solution to the following RH problem, which is obtained from the RH problem for $S$ by ignoring the exponentially small jumps. The global parametrix will give us a good approximation to $S$ for $z$ away from the endpoints of $\gamma$ and $\widetilde\gamma$.
\subsubsection*{RH problem for $P^{(\infty)}$}
\begin{itemize}
\item[(a)] $P^{(\infty)} : \mathbb{C}\setminus S^{1} \to \mathbb{C}^{2\times 2}$ is analytic.
\item[(b)] $P^{(\infty)}$ has the following jumps:
\begin{equation}
\begin{array}{l l}
P^{(\infty)}_{+}(z) = P^{(\infty)}_{-}(z)  \begin{pmatrix}
0 & 1 \\ -1 & 0
\end{pmatrix}, & \mbox{ for } z \in \gamma\cup\widetilde\gamma, \\
P^{(\infty)}_{+}(z) = P^{(\infty)}_{-}(z)e^{-n\pi i \Omega \sigma_{3}}, &\mbox{ for } z \in \Sigma_1,\\
P^{(\infty)}_{+}(z) = P^{(\infty)}_{-}(z)e^{n\pi i \Omega \sigma_{3}}, &\mbox{ for } z \in \Sigma_2.\\
\end{array}
\end{equation}
\item[(c)] $P^{(\infty)}(z) = I + \bigO(z^{-1})$ as $z \to \infty$.
\item[(d)] As $z\to \zeta$ with $\zeta = e^{\pm i L}$ or $\zeta=e^{i(\pi\pm T)}$, we have
\[
P^{(\infty)}(z)=\bigO(|z-\zeta|^{-1/4}).
\]
\end{itemize}

This problem can be explicitly solved in terms of the elliptic theta-function of the third kind. This is typical for situations where the support of the equilibrium measure consists of two disjoint intervals or arcs. Similar RH problems have been solved several times with jumps on the real line, see e.g.\ \cite{DKMVZ2, DKMVZ1, Bleher}. We follow a similar procedure here, adapted to the unit circle. 

First we apply the following transformation:
\begin{equation}\label{def Q}
Q^{(\infty)}(z) = e^{-n\pi i \frac{\Omega}{2} \sigma_{3}}P^{(\infty)}(z) \left\{ \begin{array}{l l}
e^{-n\pi i \frac{\Omega}{2}\sigma_{3}}, & |z|<1, \\
e^{n\pi i \frac{\Omega}{2}\sigma_{3}}, & |z|>1. \\
\end{array} \right.
\end{equation}
It is easy to verify that $Q^{(\infty)}$ satisfies the following RH problem.

\subsubsection*{RH problem for $Q^{(\infty)}$}
\begin{itemize}
\item[(a)] $Q^{(\infty)} : \mathbb{C}\setminus (\supp\,\mu\cup \Sigma_1) \to \mathbb{C}^{2\times 2}$ is analytic.
\item[(b)] $Q^{(\infty)}$ has the following jumps:
\begin{equation}
\begin{array}{l l}
Q^{(\infty)}_{+}(z) = Q^{(\infty)}_{-}(z)  \begin{pmatrix}
0 & 1 \\ -1 & 0
\end{pmatrix},& \mbox{ for } z \in \gamma\cup\widetilde\gamma, \\
Q^{(\infty)}_{+}(z) = Q^{(\infty)}_{-}(z)e^{-2n\pi i \Omega \sigma_{3}}, & \mbox{ for } z \in \Sigma_1.\\
\end{array}
\end{equation}
\item[(c)] $Q^{(\infty)}(z) = I + O(z^{-1})$ as $z \to \infty$.
\item[(d)] As $z\to \zeta$ with $\zeta = e^{\pm i L}$ or $\zeta=e^{i(\pi\pm T)}$, we have
\[
Q^{(\infty)}(z)=\bigO(|z-\zeta|^{-1/4}).
\]
\end{itemize}

\subsubsection*{Riemann surface and elliptic theta function}
To construct $Q$, we need to introduce quantities related to an elliptic theta-function on a Riemann surface. 
We consider the elliptic curve 
\[
X = \{ (z,w):w^{2} = R(z) \},\hspace{0.5cm} R(z) = (z-\overline{z_{0}})(z-z_{0})(z-z_{1})(z-\overline{z_{1}}),
\]
of genus one. We represent $X$ as the two-sheeted Riemann surface associated to $\sqrt{R(z)}$. We let $\sqrt{R(z)} \sim z^{2}$ as $z \to \infty$ on the first sheet and $\sqrt{R(z)} \sim -z^{2}$ as $z \to \infty$ on the second sheet.

We define cycles $A$ and $B$ on $X$ as in  Figure \ref{fig_AB}: $B$ encircles the arc $\widetilde\gamma$ in the clockwise sense on the first sheet, and $A$ encircles the arc $\Sigma_1$. The solid part in Figure \ref{fig_AB} lies on the first sheet, the dashed part on the second sheet. $A$ and $B$ form a canonical homology basis for $X$.

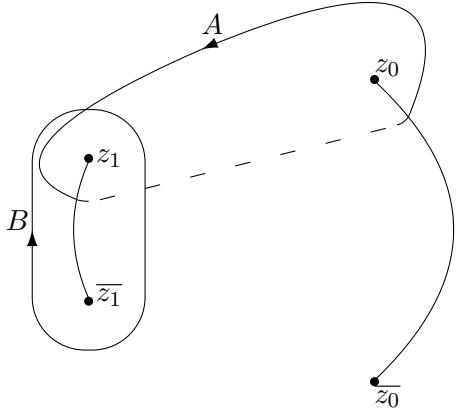
\begin{figure}[t]
    \begin{center}
    \setlength{\unitlength}{1truemm}
    \begin{picture}(100,55)(-5,10)
        %\cCircle(50,40){25}[f]
        \put(65,60){\thicklines\circle*{1.2}}
        \put(65,19.8){\thicklines\circle*{1.2}}
        \put(65,61){$z_0$} 
        \put(65,17){$\overline{z_0}$}
        \put(27,49.5){\thicklines\circle*{1.2}}
        \put(27,30.5){\thicklines\circle*{1.2}}
        \put(28,49){$z_{1}$}
        \put(28,30.5){$\overline{z_{1}}$}
        
        % to draw B
        \put(27,40){\oval(15,32)}
        \put(16,40){$B$}
        \put(19.5,40){\thicklines\vector(0,1){.0001}}
        %to draw A
        \qbezier(69.5,55)(80,80)(42,64)
        \qbezier(42,64)(10,50)(25.5,44)
        \qbezier(69.5,55)(68.7,54)(68,54)
        \qbezier(25.5,44)(26,43.8)(27,43.8)
        \put(42,66){$A$}
        \put(42,64){\thicklines\vector(-1,-0.45){.0001}}
        \multiput(28.5,44)(6,1.5){7}{\line(1,0.25){2}}
        
        %to draw the piece of circle
        \qbezier(27,49.1)(23,40)(27,30.9)
        \qbezier(64.5,19.5)(86.4,40)(64.8,60)

    \end{picture}
    \caption{The cycles $A$ and $B$.}
    \label{fig_AB}
\end{center}

\end{figure}

We define the one-form
\begin{equation}\label{def omega}\omega = \frac{c_{0}dz}{\sqrt{R(z)}},\qquad c_{0} = \frac{1}{\int_{A}\frac{dz}{\sqrt{R(z)}}} \in \mathbb{R}^{+},\end{equation} such that $\omega$ is the unique holomorphic one-form on $X$ which satisfies $\int_{A} \omega = 1$. We also let
\begin{equation} \label{tau def}
\displaystyle \tau = \int_{B} \omega = \frac{\int_{B} \frac{dz}{\sqrt{R(z)}}}{\int_{A} \frac{dz}{\sqrt{R(z)}}}  \in i\mathbb{R}^{+}.
\end{equation}
The associated theta-function of the third kind $\theta(z;\tau) = \theta(z)$ is given by
\begin{equation}\label{def theta}
\theta(z) = \sum_{m=-\infty}^{+\infty} e^{2\pi i m z}e^{\pi i m^{2}\tau}.
\end{equation}
It is an entire and even function satisfying
\begin{equation} \label{periodicity of theta}
\theta(z+1) = \theta(z) \qquad \mbox{ and } \qquad \theta(z+\tau) = e^{-2\pi i z}e^{-\pi i \tau} \theta(z) \qquad \mbox{ for all } z \in \mathbb{C}.
\end{equation}
Finally, define
\begin{equation}
u(z) = \int_{\overline{z_{1}}}^{z} \omega, 
\end{equation}
for $z$ on the Riemann surface,
where the path of integration lies in $\mathbb{C}\setminus (\overline{\gamma\cup\widetilde\gamma \cup \Sigma_1})$ and on the same sheet as $z$. Since $\oint_{C(0,R)} \omega = 0$ for all $R>1$, with $C(0,R)$ a circle of radius $R$ around $0$, on each sheet, $u$ can be seen as a single-valued and analytic function of $z\in\mathbb{C}\setminus (\overline{\gamma\cup\widetilde\gamma \cup \Sigma_1})$. For $z$ on the first sheet, it satisfies the relations
\begin{equation}
\begin{array}{l l}
\displaystyle \mbox{for } z \in \Sigma_1, & \displaystyle u_{+}(z) - u_{-}(z) = - \int_{B} \omega = -\tau, \\
\displaystyle \mbox{for } z \in \widetilde{\gamma}, & \displaystyle u_{+}(z) + u_{-}(z) = 0, \\
\displaystyle \mbox{for } z \in \gamma, & \displaystyle u_{+}(z)+u_{-}(z) = - \int_{A} \omega = -1.
\end{array}
\end{equation}

\subsubsection*{Construction of $Q^{(\infty)}$}

Now we define 
\begin{equation}f_{1}(z;c,d) = \frac{\theta(u(z)+d+c)}{\theta(u(z)+d)},\qquad f_{2}(z;c,d) = \frac{\theta(-u(z)+d+c)}{\theta(-u(z)+d)},
\end{equation} 
where we take $z$ on the first sheet. These functions are meromorphic on $\mathbb{C}\setminus (\overline{\gamma\cup\widetilde\gamma \cup \Sigma_1})$ with possible poles at the zeros of $\theta(u(z)+d)$ and $\theta(-u(z)+d)$. Furthermore, by the periodicity properties of the $\theta$-function, we have the following relations:

\begin{equation}
\begin{array}{l l}
\mbox{for } z \in \Sigma_1, & f_{1}(z;c,d)_{+} = e^{2\pi i c}f_{1}(z;c,d)_{-},\\
& f_{2}(z;c,d)_{+} = e^{-2\pi i c}f_{2}(z;c,d)_{-},\\
\mbox{for } z \in \gamma \cup \tilde{\gamma}, & f_{1}(z;c,d)_{+} = f_{2}(z;c,d)_{-},\\
& f_{2}(z;c,d)_{+} = f_{1}(z;c,d)_{-}.\\
\end{array}
\end{equation}
This implies that 
\[
F(z;c,d_1,d_2) = \begin{pmatrix}
f_1(z;c,d_1) & f_2(z;c,d_1) \\
f_1(z;c,d_2) & f_2(z;c,d_2)
\end{pmatrix}
\]
satisfies the jump relations 
\begin{align}&F_{+}(z) = F_{-}(z) e^{2\pi i c \sigma_{3}}, &\mbox{ for $z \in \Sigma_1$},\\
&F_{+}(z) = F_{-}(z) \begin{pmatrix}
0 & 1 \\ 1 & 0
\end{pmatrix},&\mbox{ for $z \in  \gamma\cup\widetilde\gamma$.}
\end{align} On the other hand, 
the function
\begin{equation}
N(z) = \begin{pmatrix}
\frac{1}{2}(\beta(z)+\beta^{-1}(z)) & \frac{1}{-2i}(\beta(z)-\beta^{-1}(z)) \\
\frac{1}{2i}(\beta(z)-\beta^{-1}(z)) & \frac{1}{2}(\beta(z)+\beta^{-1}(z))
\end{pmatrix},
\end{equation}
where 
\begin{equation}  \label{equation for beta}
\beta(z) = \left(\frac{z-\overline{z_{0}}}{z-z_{0}}\frac{z-z_{1}}{z-\overline{z_{1}}}\right)^{1/4}
\end{equation}
with branch cut on $\gamma\cup\widetilde\gamma$ and such that $\beta(z)\to 1$ as $z\to\infty$, satisfies the following RH problem, see e.g.\ \cite{Bleher} for similar situations:

\subsubsection*{RH problem for $N$}
\begin{itemize}
\item[(a)] $N : \mathbb{C}\setminus \supp\,\mu \to \mathbb{C}^{2\times 2}$ is analytic.
\item[(b)] $N$ has the following jump:
\[
N_{+}(z) = N_{-}(z)  \begin{pmatrix}
0 & 1 \\ -1 & 0
\end{pmatrix}, \hspace{0.5cm} \mbox{ for $z$ on $\gamma\cup\widetilde\gamma$.}
\]
\item[(c)] $N(z) = I + O(z^{-1})$ as $z \to \infty$.
\item[(d)] As $z\to \zeta$, $\zeta = z_{0},\overline{z_{0}}, z_{1}$ or $ \overline{z_{1}}$, we have
\[
N(z)=\bigO(|z-\zeta|^{-1/4}).
\]
\end{itemize}

Now, we take $c=-n\Omega$ and $d_1=-d_2=d$, and define $Q^{(\infty)}$ by
\begin{multline} \label{F RHP}
Q^{(\infty)}(z) = \begin{pmatrix}
\frac{\theta(u_{\infty} + d -n\Omega)}{\theta(u_{\infty} + d)} & 0 \\
0 & \frac{\theta(u_{\infty} + d +n\Omega)}{\theta(u_{\infty} + d)}
\end{pmatrix}^{-1} \times \\
\begin{pmatrix}
\displaystyle N_{11}(z)F_{11}(z;-n\Omega,d,-d)& \displaystyle N_{12}(z)F_{12}(z;-n\Omega,d,-d)\\
\displaystyle N_{21}(z)F_{21}(z;-n\Omega,d,-d)& \displaystyle N_{22}(z)F_{22}(z;-n\Omega,d,-d)
\end{pmatrix},
\end{multline}
where we define $u_{\infty}=\lim_{z\to\infty}u(z)$. Combining the jump relations for $F$ and $N$, it is straightforward to verify that $Q^{(\infty)}$ satisfies the RH problem for $Q^{(\infty)}$, except for the possible problem that it may have poles at the zeros of the functions $\theta(u(z)-d)$ and $\theta(u(z)+d)$. 
We will exploit the freedom we have to choose the value of $d$ to ensure that the zeros in the denominator are cancelled out by zeros in the numerator, so that $Q^{(\infty)}$ is analytic in $\mathbb C\setminus(\supp\,\mu\cup\Sigma_1)$.

For $w \in \mathbb{C}\setminus \supp\,\mu$, we denote by $w^{(1)}$ the representation of $w$ on the first sheet of $X$, and $w^{(2)}$ for the one on the second sheet. We define
\begin{equation}\label{z_star}
z_{\star} = \frac{z_{0}\overline{z_{1}}-\overline{z_{0}}z_{1}}{(z_{0}-\overline{z_{0}})-(z_{1}-\overline{z_{1}})}.
\end{equation}
This is the projection of $z_{1}$ along the vector $z_{1}-z_{0}$ on the real axis $\mathbb{R}$. So if $0 < \sin T < \sin L$, $z_{\star} \in ]-\infty,-1[$, if $\sin T = \sin L$, $z_{\star} = \infty$, and if $\sin L < \sin T$, $z_{\star} \in ]1,+\infty[$.

\begin{proposition}\label{proposition for Pinf}
Let $d = -K + \int_{\overline{z_{1}}}^{z_{\star}^{(1)}} \omega$, where $K = \frac{1}{2} + \frac{\tau}{2}$.
Then $Q^{(\infty)}$ defined by equation \eqref{F RHP} is the solution of the RH problem for $Q^{(\infty)}$. Moreover, $u_{\infty} + d \equiv 0 \mod 1$.
\end{proposition}

\begin{proof}
First of all, we identify the zeros of the functions $\beta(z) \pm \beta^{-1}(z)$ on the Riemann surface.  One of those functions vanishes at a point $z$ on the Riemann surface if and only if
\begin{equation}
 \beta(z)^{2} \pm 1 = 0 \Leftrightarrow \beta^{4}(z) = 1 \Leftrightarrow z = z_{\star}.
\end{equation}

From the definition of $\beta$ on the first sheet, we have that $\beta(z_{\star}^{(1)})>0$. So $\beta(z)-\beta^{-1}(z)$ vanishes only at $z_{\star}^{(1)}$ and $\beta(z)+\beta^{-1}(z)$ does not vanish on the first sheet. 
Since $X$ is of genus one, we know that $u(z)$ is a bijection from the Riemann surface $X$ to the Jacobi variety $\mathbb{C}/ \Lambda$, $\Lambda = \{ n + \tau m , n,m \in \mathbb{Z}\}$. Therefore, since $\theta(K)=0$,
\[
\theta(u(z)-d) = \theta \left( \int_{\overline{z_{1}}}^{z} \omega - \int_{\overline{z_{1}}}^{z_{\star}^{(1)}} \omega + K \right)
\]
vanishes only at $z = z_{\star}^{(1)}$. A similar argument shows that
\[
\theta(u(z)+d) = \theta \left( \int_{\overline{z_{1}}}^{z} \omega - \int_{\overline{z_{1}}}^{z_{\star}^{(2)}}\omega - K \right)
\]
vanishes only at $z = z_{\star}^{(2)}$. By \eqref{F RHP}, it follows that $Q^{(\infty)}$ has no poles in $\mathbb C\setminus(\supp\,\mu\cup\Sigma_1)$. Hence it solves the RH problem for $Q^{(\infty)}$. 

\medskip

Now, note that $\beta^{2}(z)-1 = \left(\frac{(z-\overline{z_{0}})(z-z_{1})}{(z-z_{0})(z-\overline{z_{1}})}\right)^{1/2}-1$ is meromorphic on $X$, has simple zeros at $z_{\star}^{(1)}$ and $\infty^{(1)}$, and simple poles at $z_{0}$ and $\overline{z_{1}}$. By the Abel theorem, we have
\begin{equation}\label{id abel}
\int_{\overline{z_{1}}}^{\infty^{(1)}} \omega + \int_{z_{0}}^{z_{\star}^{(1)}} \omega \equiv 0 \mod \Lambda.
\end{equation}
By the choice of cycles and the definition of $\omega$, there exist $n, m \in \mathbb{Z}$ such that
\begin{equation}
\int_{z_{\star}^{(1)}}^{-1_{-}^{(1)}} \omega + \int_{\infty^{(1)}}^{-1_{-}^{(1)}} \omega = \frac{1}{2} + n + m\tau,
\end{equation}
where $-1_{-}^{(1)}$ means that we start the integration path from $-1$ on the side $|z| > 1$ of the first sheet. For $z<-1$ on the first sheet, $\omega>0$ so that $m=0$ if $z_{\star} < -1$ or if $z_{\star}=\infty$. 
From the definitions of $u_\infty$, $d$, and $K$, and by \eqref{id abel} and the choice of cycles, it is then straightforward to see that 
\[u_\infty+d=
\int_{\overline{z_1}}^{z_*^{(1)}}\omega +\int_{z_*^{(1)}}^{z_0}\omega - n -\frac{1}{2}-\frac{\tau}{2}\equiv 0\mod 1.\]
If $z_{\star}>1$, this can be shown in a similar way.
\end{proof}

Summarizing, by \eqref{def Q} and \eqref{F RHP}, we have the following expression for $P^{(\infty)}$:
\begin{multline}\label{final expression Pinf}
P^{(\infty)}(z) = e^{n\pi i \frac{\Omega}{2}\sigma_{3}} \begin{pmatrix}
\frac{1}{2}(\beta(z)+\beta^{-1}(z))\Theta_{11}(z) & \frac{1}{-2i}(\beta(z)-\beta^{-1}(z))\Theta_{12}(z) \\
\frac{1}{2i}(\beta(z)-\beta^{-1}(z))\Theta_{21}(z) & \frac{1}{2}(\beta(z)+\beta^{-1}(z))\Theta_{22}(z)
\end{pmatrix}\\
\times\ \begin{cases}
e^{n\pi i \frac{\Omega}{2}\sigma_{3}}, & |z|<1, \\
e^{-n\pi i \frac{\Omega}{2}\sigma_{3}}, & |z|>1 ,\end{cases}
\end{multline}
with
\begin{equation} \label{ratio theta functions}
\begin{array}{l l}
\displaystyle \Theta_{11}(z) = \frac{\theta(0)}{\theta(n\Omega)} \frac{\theta(u(z)+d-n\Omega)}{\theta(u(z)+d)}, & \qquad \displaystyle \Theta_{12}(z) = \frac{\theta(0)}{\theta(n\Omega)}\frac{\theta(u(z)-d+n\Omega)}{\theta(u(z)-d)}, \\
\displaystyle \Theta_{21}(z) = \frac{\theta(0)}{\theta(n\Omega)}\frac{\theta(u(z)-d-n\Omega)}{\theta(u(z)-d)}, & \qquad \displaystyle \Theta_{22}(z) = \frac{\theta(0)}{\theta(n\Omega)}\frac{\theta(u(z)+d+n\Omega)}{\theta(u(z)+d)}.\\
\end{array}
\end{equation}
In the above formulas involving $u(z)$, $z$ is taken on the first sheet.

\subsection{Local parametrices near $z_{0} = e^{iL}$ and $\overline{z_0}=e^{-iL}$}
Near $z_0$ and $\overline{z_0}$, we want to construct a function which has exactly the same jump conditions than $S$, and which matches with the global parametrix on the boundaries of small fixed disks of radius $r > 0$, $D(z_0,r)$ or  $D(\overline{z_0},r)$, around $z_0$ and $\overline{z_0}$. More precisely, near $z_0$, we want to have the following conditions.

\subsubsection*{RH problem for $P$}

\begin{itemize}
\item[(a)] $P : D(z_{0},r) \setminus (S^{1} \cup \gamma_{+} \cup \gamma_{-}) \to \mathbb{C}^{2\times 2}$ is analytic.
\item[(b)] For $z\in D(z_{0},r) \cap (S^{1} \cup \gamma_{+} \cup \gamma_{-})$, $P$ satisfies the jump conditions
\begin{equation}\label{jumps P}
\begin{array}{l l}
P_{+}(z) = P_{-}(z) \begin{pmatrix}
0 & 1 \\ -1 & 0
\end{pmatrix}, & \mbox{ on } \gamma, \\

P_{+}(z) = P_{-}(z) \begin{pmatrix}
 e^{-n\pi i \Omega} & e^{n(\phi(z)-x)} \\
 0 & e^{n\pi i \Omega}
\end{pmatrix}, & \mbox{ on } \Sigma_1, \\

P_{+}(z) = P_{-}(z) \begin{pmatrix}
 1 & 0  \\ e^{-n\phi(z)}e^{-n\pi i \Omega} & 1
\end{pmatrix}, & \mbox{ on } \gamma_{+}, \\

P_{+}(z) = P_{-}(z) \begin{pmatrix}
 1 & 0  \\ e^{-n\phi(z)}e^{n\pi i \Omega} & 1
\end{pmatrix}, & \mbox{ on } \gamma_{-}. \\
\end{array}
\end{equation}
\item[(c)] For $z \in \partial D (z_{0},r)$, we have
\begin{equation}\label{matching P z0} P(z) = \left(I + \bigO(n^{-1})\right) P^{(\infty)}(z),\qquad \mbox{ as $n \to \infty$.}
\end{equation}
\item[(d)] As $z$ tends to $z_{0}$, the behaviour of $P$ is
\begin{equation}\label{P local}
\begin{array}{l l}
P(z) = \bigO(\log|z-z_0|).
\end{array}
\end{equation}
\end{itemize}

We can solve this RH problem in the same way as done in \cite{ChCl} in terms of Bessel functions. The construction in \cite{ChCl} is an adaptation of the standard Bessel parametrix construction from \cite{Kuijlaars2}.

Define $\Psi$ by
\begin{equation}\label{Psi explicit}
\Psi(\zeta)=\begin{cases}
\begin{pmatrix}
I_{0}(2\zeta^{\frac{1}{2}}) & \frac{i}{\pi} K_{0}(2\zeta^{\frac{1}{2}}) \\
2\pi i \zeta^{\frac{1}{2}} I_{0}^{\prime}(2\zeta^{\frac{1}{2}}) & -2\zeta^{\frac{1}{2}} K_{0}^{\prime}(2\zeta^{\frac{1}{2}})
\end{pmatrix}, & |\arg \zeta | < \frac{2\pi}{3}, \\

\begin{pmatrix}
\frac{1}{2} H_{0}^{(1)}(2(-\zeta)^{\frac{1}{2}}) & \frac{1}{2} H_{0}^{(2)}(2(-\zeta)^{\frac{1}{2}}) \\
\pi \zeta^{\frac{1}{2}} \left( H_{0}^{(1)} \right)^{\prime} (2(-\zeta)^{\frac{1}{2}}) & \pi \zeta^{\frac{1}{2}} \left( H_{0}^{(2)} \right)^{\prime} (2(-\zeta)^{\frac{1}{2}})
\end{pmatrix}, & \frac{2\pi}{3} < \arg \zeta < \pi, \\

\begin{pmatrix}
\frac{1}{2} H_{0}^{(2)}(2(-\zeta)^{\frac{1}{2}}) & -\frac{1}{2} H_{0}^{(1)}(2(-\zeta)^{\frac{1}{2}}) \\
-\pi \zeta^{\frac{1}{2}} \left( H_{0}^{(2)} \right)^{\prime} (2(-\zeta)^{\frac{1}{2}}) & \pi \zeta^{\frac{1}{2}} \left( H_{0}^{(1)} \right)^{\prime} (2(-\zeta)^{\frac{1}{2}})
\end{pmatrix}, & -\pi < \arg \zeta < -\frac{2\pi}{3},
\end{cases}
\end{equation}
where $H_0^{(1)}$ and $H_0^{(2)}$ are the Hankel functions of the first and second kind, and $I_0$ and $K_0$ are the modified Bessel functions of the first and second kind. Next, let
\begin{equation}\label{def hatPsi}
\widehat{\Psi}(\zeta) = \left( I + A(\zeta) \right) \Psi(\zeta), 
\end{equation}
where $A(\zeta)$ is a nilpotent matrix which is needed to take care of the jump on $\Sigma_1$,
\begin{equation}\label{def A}
A(\zeta) = e^{-nx} F(\zeta) \begin{pmatrix}
0 & - \frac{1}{2\pi i} \log(-\zeta) \\ 0 & 0 \\
\end{pmatrix} F^{-1}(\zeta),
\end{equation}
with the branch cut of $\log(-\zeta)$ on $\mathbb R^+$,
and $F$ is the entire function given by
\begin{equation}\label{F}
F(\zeta) = \left\{ \begin{array}{l l}
\Psi(\zeta) \begin{pmatrix}
1 & -\frac{1}{2\pi i} \log \zeta \\
0 & 1
\end{pmatrix} & | \arg \zeta | < \frac{2\pi}{3}, \\

\Psi(\zeta) \begin{pmatrix}
1 & 0 \\
1 & 1
\end{pmatrix} \begin{pmatrix}
1 & -\frac{1}{2\pi i} \log \zeta \\
0 & 1
\end{pmatrix} & \frac{2\pi}{3} < \arg \zeta < \pi, \\

\Psi(z) \begin{pmatrix}
1 & 0 \\
-1 & 1
\end{pmatrix} \begin{pmatrix}
1 & -\frac{1}{2\pi i} \log \zeta \\
0 & 1
\end{pmatrix} & -\pi < \arg \zeta < -\frac{2\pi}{3}. \\

\end{array} \right.
\end{equation}
It was shown in \cite{ChCl} that $\widehat\Psi$
is the solution of the following RH problem.

\subsubsection*{RH problem for $\widehat{\Psi}$}

\begin{itemize}
\item[(a)] $\widehat{\Psi} : \mathbb{C} \setminus \Sigma_{\widehat{\Psi}} \to \mathbb{C}^{2\times 2}$ is analytic, where $\Sigma_{\widehat{\Psi}}$ is shown in Figure \ref{figPsiHat}.
\item[(b)] $\widehat{\Psi}$ satisfies the jump conditions
\begin{equation}\label{jumps hatPsi}
\begin{array}{l l}
\widehat{\Psi}_{+}(\zeta) = \widehat{\Psi}_{-}(\zeta) \begin{pmatrix}
0 & 1 \\ -1 & 0
\end{pmatrix}, & \mbox{ on } \mathbb{R}^{-}, \\

\widehat{\Psi}_{+}(\zeta) = \widehat{\Psi}_{-}(\zeta) \begin{pmatrix}
1 & e^{-nx} \\ 0 & 1
\end{pmatrix}, & \mbox{ on } \mathbb{R}^{+}, \\

\widehat{\Psi}_{+}(\zeta) = \widehat{\Psi}_{-}(\zeta) \begin{pmatrix}
 1 & 0  \\ 1 & 1
\end{pmatrix}, & \mbox{ on } e^{\frac{2\pi}{3}i}\mathbb{R}^{+}, \\

\widehat{\Psi}_{+}(\zeta) = \widehat{\Psi}_{-}(\zeta) \begin{pmatrix}
 1 & 0 \\ 1 & 1
\end{pmatrix}, & \mbox{ on } e^{-\frac{2\pi}{3}i} \mathbb{R}^{+}.
\end{array}
\end{equation}
\item[(c)] As $\zeta \to \infty$, $\zeta \notin \Sigma_{\Psi}$, we have
\begin{equation}
\widehat{\Psi}(\zeta) = \left( I + A(\zeta) \right)\left( 2\pi \zeta^{\frac{1}{2}} \right)^{-\frac{\sigma_{3}}{2}} \frac{1}{\sqrt{2}} \begin{pmatrix}1&i\\i&1\end{pmatrix}\left(
I+\bigO (\zeta^{-\frac{1}{2}})\right) e^{2\zeta^{\frac{1}{2}}\sigma_{3}}.
\end{equation}
\item[(d)] As $\zeta$ tends to 0, the behaviour of $\widehat{\Psi}(\zeta)$ is
\begin{equation}\label{Psi local}
\widehat{\Psi}(\zeta) = \bigO (\log |\zeta| ).
\end{equation}
\end{itemize}

\begin{figure}[t]
    \begin{center}
    \setlength{\unitlength}{1truemm}
    \begin{picture}(100,55)(-5,10)
        \put(50,40){\line(1,0){30}}
        \put(50,40){\line(-1,0){30}}
        \put(50,39.8){\thicklines\circle*{1.2}}
        \put(50,40){\line(-0.5,0.866){18}}
        \put(50,40){\line(-0.5,-0.866){18}}
        \qbezier(53,40)(52,43)(48.5,42.598)
        \put(53,43){$\frac{2\pi}{3}$}
        \put(50.3,36.8){$0$}
        \put(65,39.9){\thicklines\vector(1,0){.0001}}
        \put(35,39.9){\thicklines\vector(1,0){.0001}}
        \put(41,55.588){\thicklines\vector(0.5,-0.866){.0001}}
        \put(41,24.412){\thicklines\vector(0.5,0.866){.0001}}
        %\put(60,60){$\Sigma_{\widehat{\Psi}}$}
    \end{picture}
    \caption{The jump contour for $\widehat{\Psi}$ and $\Phi$.}
    \label{figPsiHat}
\end{center}
\end{figure}
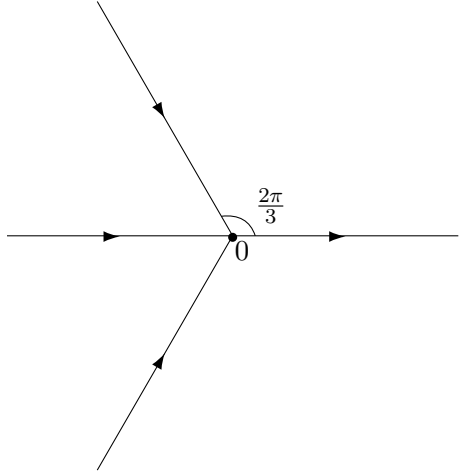

Then, we define $\zeta:D(z_0,r)\to\mathbb C$ such that
$\zeta(z) = \frac{1}{16}\phi(z)^{2}$. One can check that $\zeta$ is a conformal map which maps $z_0$ to $0$.
We can now fix the lenses $\gamma_+$ and $\gamma_-$ such that $\zeta$ maps them to the half-lines $e^{\frac{2\pi}{3}i} \mathbb{R}^{+}$ and $e^{- \frac{2\pi}{3}i} \mathbb{R}^{+}$.

We can now construct the local parametrix as follows,
\begin{equation}\label{P local para Bessel}
P(z) = E(z) \widehat{\Psi}(n^{2}\zeta(z))e^{-\frac{n}{2}\phi(z)\sigma_{3}} \times \left\{ \begin{array}{l l}
e^{-n\pi i \frac{\Omega}{2}\sigma_{3}}, & |z|<1, \\
e^{n\pi i \frac{\Omega}{2}\sigma_{3}}, & |z|>1, \\
\end{array} \right\}
\end{equation}
where 
\begin{equation} \label{E in Bessel}
E(z) = P^{(\infty)}(z) \times \left\{ \begin{array}{l l}
e^{n\pi i \frac{\Omega}{2}\sigma_{3}}, & |z|<1, \\
e^{-n\pi i \frac{\Omega}{2}\sigma_{3}}, & |z|>1, \\
\end{array} \right\} \frac{1}{\sqrt{2}} \begin{pmatrix}
1 & -i \\ -i & 1
\end{pmatrix} \left( \frac{1}{2}n\pi \phi(z) \right)^{\sigma_{3}/2},
\end{equation}
where the principal branch is chosen for $(\cdot)^{1/2}$, in such a way that $E$ is analytic in $D(z_{0},r)$. Then, $P$ solves the RH problem for $P$.

The local parametrix near $\overline{z_0}$ can simply be constructed by symmetry. For $z \in D(\overline{z_{0}},r)$, we have
\[
P(z)=\overline{P(\overline{z})}.
\]

\subsection{Local parametrix near ${z_{1}} = e^{i(\pi-T)}$ and $\overline{z_{1}} = e^{i(\pi+T)}$}
Near $z_1$ and $\overline{z_1}$, we again want to construct a function $P$ which has the same jump relations than $S$, and which matches with the global parametrix.

In $D(\overline{z_1},r)$, we need the following RH conditions.
\subsubsection*{RH problem for $P$}

\begin{itemize}
\item[(a)] $P : D(\overline{z_{1}},r) \setminus (S^{1} \cup \widetilde{\gamma}_{+} \cup \widetilde{\gamma}_{-}) \to \mathbb{C}^{2\times 2}$ is analytic.
\item[(b)] For $z\in D(\overline{z_{1}},r) \cap (S^{1} \cup \widetilde{\gamma}_{+} \cup \widetilde{\gamma}_{-})$, $P$ satisfies the jump conditions
\begin{equation}\label{jumps P2}
\begin{array}{l l}
P_{+}(z) = P_{-}(z) \begin{pmatrix}
0 & 1 \\ -1 & 0
\end{pmatrix}, & \mbox{ on } \widetilde{\gamma}, \\

P_{+}(z) = P_{-}(z) \begin{pmatrix}
 e^{n\pi i \Omega} & e^{n\widetilde{\phi}(z)} \\
 0 & e^{-n\pi i \Omega}
\end{pmatrix}, & \mbox{ on } \Sigma_2, \\

P_{+}(z) = P_{-}(z) \begin{pmatrix}
 1 & 0  \\ e^{-n\widetilde{\phi}(z)}e^{n\pi i \Omega} & 1
\end{pmatrix}, & \mbox{ on } \widetilde{\gamma}_{+}, \\

P_{+}(z) = P_{-}(z) \begin{pmatrix}
 1 & 0  \\ e^{-n\widetilde{\phi}(z)}e^{-n\pi i \Omega} & 1
\end{pmatrix}, & \mbox{ on } \widetilde{\gamma}_{-}. \\
\end{array}
\end{equation}
\item[(c)] For $z \in \partial D (\overline{z_{1}},r)$, we have
\begin{equation}\label{matching P z1} P(z) = \left(I + \bigO(n^{-1})\right) P^{(\infty)}(z),\qquad \mbox{ as $n \to \infty$.}
\end{equation}
\item[(d)] As $z$ tends to $\overline{z_{1}}$, the behaviour of $P$ is
\begin{equation}\label{P local2}
\begin{array}{l l}
P(z) = \bigO(1).
\end{array}
\end{equation}
\end{itemize}

We will construct the solution of this RH problem in terms of the Airy function.
For this, we use the following well-known model RH problem, see e.g.\ \cite{DKMVZ2, DKMVZ1, Bleher}.
\subsubsection*{Airy model RH problem}

\begin{itemize}
\item[(a)] $\Phi : \mathbb{C} \setminus \Sigma_{\Phi} \to \mathbb{C}^{2\times 2}$ is analytic, with $\Sigma_\Phi$ as in Figure \ref{figPsiHat}.
\item[(b)] For $z\in\Sigma_\Phi$, $\Phi$ satisfies the jump conditions
\begin{equation}\label{jumps P3}
\begin{array}{l l}
\Phi_{+}(z) = \Phi_{-}(z) \begin{pmatrix}
0 & 1 \\ -1 & 0
\end{pmatrix}, & \mbox{ on } \mathbb{R}^{-}, \\

\Phi_{+}(z) = \Phi_{-}(z) \begin{pmatrix}
 1 & 1 \\
 0 & 1
\end{pmatrix}, & \mbox{ on } \mathbb{R}^{+}, \\

\Phi_{+}(z) = \Phi_{-}(z) \begin{pmatrix}
 1 & 0  \\ 1 & 1
\end{pmatrix}, & \mbox{ on } e^{ \frac{2\pi}{3} i }  \mathbb{R}^{+} , \\

\Phi_{+}(z) = \Phi_{-}(z) \begin{pmatrix}
 1 & 0  \\ 1 & 1
\end{pmatrix}, & \mbox{ on }e^{ -\frac{2\pi}{3} i }\mathbb{R}^{+} . \\
\end{array}
\end{equation}
\item[(c)] As $z \to \infty$, we have
\begin{equation} 
\Phi(z) = \frac{1}{2\sqrt{\pi}}z^{-\frac{1}{4}\sigma_{3}} \begin{pmatrix}
1 & i \\ -1 & i
\end{pmatrix} \left( I + \bigO (z^{-3/2})\right) e^{-\frac{2}{3}z^{3/2}\sigma_{3}}.
\end{equation}
\item[(d)] As $z$ tends to $0$, the behaviour of $\Phi$ is
\begin{equation}\label{P local3}
\begin{array}{l l}
\Phi(z) = \bigO(1).
\end{array}
\end{equation}
\end{itemize}

The unique solution of this RH problem is given by
\begin{equation}\label{sol Airy}
\Phi(z) = \left\{ \begin{array}{l l}
\begin{pmatrix}
y_{0}(z) & -y_{2}(z) \\ 
y_{0}^{\prime}(z) & -y_{2}^{\prime}(z)
\end{pmatrix}, & \mbox{for } 0 < \arg z < \frac{2\pi}{3}, \\
\begin{pmatrix}
-y_{1}(z) & -y_{2}(z) \\ 
-y_{1}^{\prime}(z) & -y_{2}^{\prime}(z)
\end{pmatrix}, & \mbox{for } \frac{2\pi}{3} < \arg z < \pi, \\
\begin{pmatrix}
-y_{2}(z) & y_{1}(z) \\ 
-y_{2}^{\prime}(z) & y_{1}^{\prime}(z)
\end{pmatrix}, & \mbox{for } -\pi < \arg z < -\frac{2\pi}{3}, \\
\begin{pmatrix}
y_{0}(z) & y_{1}(z) \\ 
y_{0}^{\prime}(z) & y_{1}^{\prime}(z)
\end{pmatrix}, & \mbox{for } -\frac{2\pi}{3} < \arg z < 0, \\
\end{array} \right.
\end{equation}
where $y_{0}(z) = \mbox{Ai}(z)$, $y_{1}(z) = e^{\frac{2\pi}{3}i} \mbox{Ai}(e^{\frac{2\pi}{3}i}z)$ and $y_{2}(z) = e^{-\frac{2\pi}{3}i}\mbox{Ai}(e^{-\frac{2\pi}{3}i}z)$.

\subsubsection{Construction of the local parametrix}

We construct $P$ using the solution $\Phi$ to the Airy model RH problem. We take $P$ of the form
\begin{equation}\label{def P Airy}
P(z) = E(z) \Phi (n^{2/3} \widetilde\zeta(z)) e^{-\frac{n}{2}\widetilde{\phi}(z)\sigma_{3}} \times  \begin{cases}
e^{n\pi i \frac{\Omega}{2}\sigma_{3}}, & |z| < 1 \\
e^{-n\pi i \frac{\Omega}{2}\sigma_{3}}, & |z| > 1 \\
\end{cases}.
\end{equation}
Here $\widetilde\zeta$ is a conformal map near $\overline{z_1}$, defined by $\widetilde\zeta(z) = \left( -\frac{3}{4}\widetilde{\phi}(z) \right)^{\frac{2}{3}}$. 
Then, we have $\widetilde\zeta(\widetilde{\gamma}\cap D(\overline{z_{1}},r)) \subset \mathbb{R}^{-}$ and $\widetilde\zeta(\Sigma_2\cap D(\overline{z_{1}},r)) \subset \mathbb{R}^{+}$. We had some freedom in the choice of the lenses, we use it now, by choosing $\widetilde{\gamma}_{+}$ and $\widetilde{\gamma}_{-}$ such that 
\[
\begin{array}{l}
\displaystyle \widetilde\zeta( \widetilde{\gamma}_{+} \cap D(\overline{z_{1}},r)) \subset e^{ \frac{2\pi}{3} i }\mathbb{R}^{+}, \\
\displaystyle \widetilde\zeta( \widetilde{\gamma}_{-} \cap D(\overline{z_{1}},r)) \subset e^{ -\frac{2\pi}{3} i } \mathbb{R}^{+}.\\
\end{array}
\]
In this way, if $E$ is analytic in $D(\overline{z_{1}},r)$, $P$ has its jumps precisely on the jump contour for $S$. Using \eqref{def P Airy}, one verifies moreover that $P$ has exactly the same jumps than $S$ inside $D(\overline{z_{1}},r)$.

In order to have the matching condition \eqref{matching P z1}, we now define $E$ as
\begin{equation}\label{E in airy}
E(z) = P^{(\infty)}(z) \times \left\{ \begin{array}{l l}
e^{-n\pi i \frac{\Omega}{2}\sigma_{3}}, & |z| < 1 \\
e^{n\pi i \frac{\Omega}{2}\sigma_{3}}, & |z| > 1 \\
\end{array} \right\}\times \left[ \frac{1}{2\sqrt{\pi}} n^{-\frac{1}{6}\sigma_{3}} \widetilde\zeta(z)^{-\frac{1}{4}\sigma_{3}} \begin{pmatrix}
1 & i \\ -1 & i
\end{pmatrix} \right]^{-1},
\end{equation} with  $\widetilde\zeta(z)^{-1/4\sigma_3}$ analytic except for $\widetilde\zeta(z) \in\mathbb{R}^{-}$, or $z\in\widetilde\gamma$. Using the jump relations satisfied by $P^{(\infty)}$, it is straightforward to check that $E$ is analytic in $D(\overline{z_{1}},r))$. In this way, $P$ satisfies the RH conditions needed near $\overline{z_1}$. In $D({z_{1}},r)$, the local parametrix is directly constructed as
$P(z)=\overline{P(\overline{z})}$.
\subsection{Final transformation $S\mapsto R$}\label{section: R}

Define
\begin{equation}\label{R}
R(z)=\begin{cases}
S(z)P(z)^{-1},&z\in D(z_{0},r)\cup D(\overline{z_{0}},r)\cup D(z_{1},r)\cup D(\overline{z_{1}},r),\\
S(z)P^{(\infty)}(z)^{-1}, & \mbox{elsewhere}.
\end{cases}
\end{equation}
$P$ was constructed in such a way that it has exactly the same jump relations as $S$ in the four small disks, and as a consequence $R$ has no jumps at all inside those disks. Moreover, from the local behaviour of $S$ and $P$ near $z_0$, $\overline{z_{0}}$, $z_{1}$, $\overline{z_{1}}$, it follows that $R$ is analytic at these four points.

If we orient the circles in the clockwise sense, the jump matrices for $R$ are given by
\begin{equation}\label{jump R circles}
R_-(z)^{-1}R_+(z)=P(z)P^{(\infty)}(z)^{-1},\quad\mbox{ for $z$ on the four circles,}
\end{equation}
and by 
\begin{equation}\label{jump R lenses}
R_-(z)^{-1}R_+(z)=P^{(\infty)}(z)S_-(z)^{-1}S_+(z)P^{(\infty)}(z)^{-1},\end{equation} 
for $z\in\Sigma_S \setminus \supp\,\mu$ and $z$ outside the four disks, where $\Sigma_S$ is the jump contour for $S$. Notice that $R$ is analytic on the part of $\gamma\cup\widetilde\gamma$ outside the four disks. By the fact that the jump matrices for $S$ converge to the identity matrix rapidly outside the four disks, by the uniform boundedness of $P^{(\infty)}$ away from the four edge points, and by the matching conditions \eqref{matching P z0} and \eqref{matching P z1}, the jump matrices for $R$ become small as $n\to\infty$.

In conclusion, we have the following RH problem for $R$:
\subsubsection*{RH problem for $R$}

\begin{figure}[t]
    \begin{center}
    \setlength{\unitlength}{1truemm}
    \begin{picture}(100,55)(0,10)
    % centre du gros cercle en (50,40)
        \cCircle(65,60){6}[f]
        \cCircle(65,20){6}[f]
        \cCircle(30,55){6}[f]
        \cCircle(30,25){6}[f]
        \qbezier(69,55)(75,40)(69,24.8)
        \qbezier(71,59)(100,40)(71,21)
%        \qbezier(26.125,47.415)(33.8525,72.2961)(58.3936,63.5488) %le point du milieu est l'intersection des tangentes au cercle
%        \qbezier(26.125,32.585)(33.8525,7.7039)(58.3936,16.4512)
        \qbezier(34.15,59.33)(46,69.0665)(59.8,63)
        \qbezier(34.15,20.67)(46,10.9335)(59.8,17)
        \qbezier(25.5,51)(10,40)(25.5,29)
        \qbezier(29,49.1)(25,40)(29,30.9)
        
        \put(67,65.8){\thicklines\vector(1,0){.0001}}
        \put(63.8,14){\thicklines\vector(-1,0){.0001}}
        \put(72.07,41.5){\thicklines\vector(0,1){.0001}}
        \put(85.5,41.5){\thicklines\vector(0,1){.0001}}
        \put(49.4,65.2){\thicklines\vector(-1,0){.0001}}
        \put(50.8,14.8){\thicklines\vector(1,0){.0001}}
        \put(31.1,60.9){\thicklines\vector(1,0){.0001}}
        \put(29,18.9){\thicklines\vector(-1,0){.0001}}
        \put(27,38.8){\thicklines\vector(0,-1){.0001}}
        \put(17.9,38.8){\thicklines\vector(0,-1){.0001}}
        
        %\put(40,46){$\gamma^{c}$}
        %\put(72.2,40){$\gamma_{+}$}
        %\put(86.2,40){$\gamma_{-}$}
        
        \put(64,59){$z_{0}$} 
        \put(64,19){$\overline{z_{0}}$}
        \put(29,54){$z_{1}$} 
        \put(29,24){$\overline{z_{1}}$}
        
        \put(48,68){$\Sigma_{R}$}
    \end{picture}
    \caption{The jump contour for $R$.}
    \label{figure: contour R}
\end{center}
\end{figure}
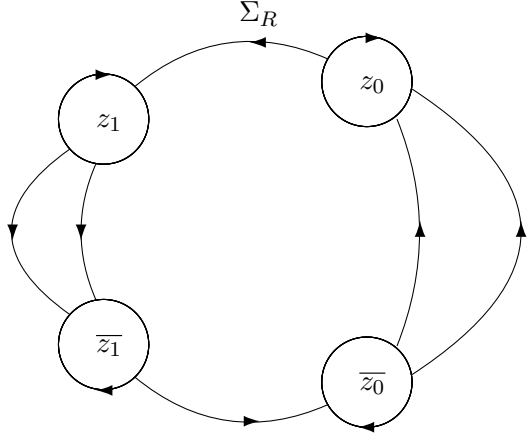

\begin{itemize}
\item[(a)] $R : \mathbb{C} \setminus \Sigma_{R} \to \mathbb{C}^{2\times 2}$ is analytic, where $\Sigma_R$ consists of the four circles and the part of $\Sigma_S \setminus \supp\,\mu$ outside the four disks, as in Figure \ref{figure: contour R}.
\item[(b)] As $n\to\infty$ with $0<x<x_c$, $R$ satisfies the jump conditions
\begin{equation}
\begin{array}{l l}
\displaystyle R_{+}(z) = R_{-}(z) \left(I + \bigO(n^{-1})\right), & \mbox{ for $z$ on the four circles}, \\[0.2cm]

\displaystyle R_{+}(z) = R_{-}(z) (I+\bigO(e^{-cn})), & \mbox{ for } z \mbox{ elsewhere on } \Sigma_{R},

\end{array}
\end{equation}
with $c > 0$ a constant independent of $n$.
\item[(c)] $R(z) = I + \bigO(z^{-1})$ as $z \to \infty$.
\end{itemize}
From the standard theory for small-norm RH problems \cite{DKMVZ2, DKMVZ1}, it follows that \begin{equation}R(z) = I + \bigO(n^{-1}),\qquad R^{\prime}(z) = \bigO(n^{-1}),\qquad n\to\infty,\label{as R}
\end{equation}
uniformly for $z \in \mathbb{C}\setminus \Sigma_{R}$. The asymptotics \eqref{as R} hold uniformly as long as $x$ lies in a compact subset of $(0,x_c)$, or if $s$ is in a region of the type $e^{-(x_{c}-\delta)n} \leq s \leq e^{-\delta n}$, with $\delta > 0$ a fixed constant.

\subsection{Asymptotics for the orthogonal polynomials}
\label{section: as phin}
We can now use the RH analysis to obtain large $n$ asymptotics for the orthogonal polynomials $\phi_n=Y_{11}$, valid for $\delta<x<x_c-\delta$, for any $\delta>0$.

\subsubsection{The outer region}

For any $\epsilon>0$, if $|z|<1-\epsilon$ or $|z|>1+\epsilon$, we can take the lenses sufficiently close to the unit circle and the disks for the local parametrices sufficiently small, such that $z$ lies in the region outside the lenses and outside the disks.

Inverting the transformations $Y\mapsto T\mapsto S\mapsto R$, we can express $Y$ in terms of $R$, and obtain the identity
\begin{equation}
Y(z) = e^{\frac{n\pi i}{2}\sigma_{3}} e^{-\frac{n\ell}{2}\sigma_{3}}R(z)P^{(\infty)}(z)e^{-\frac{n\pi i}{2}\sigma_{3}}e^{\frac{n\ell}{2}\sigma_{3}}e^{ng(z)\sigma_{3}}.
\end{equation} By \eqref{final expression Pinf}, this leads to
\begin{equation}
\phi_{n}(z) = e^{ng(z)} \left( P_{11}^{(\infty)}(z)R_{11}(z) + P_{21}^{(\infty)}(z)R_{12}(z) \right).
\end{equation}
Since $P^{(\infty)}(z)$ is uniformly bounded in $n$ for $z$ in this region and by \eqref{as R}, we get
\begin{equation}\label{Y11 near inf and 0}
\phi_{n}(z) = e^{ng(z)} \left( P_{11}^{(\infty)}(z)+ \bigO(n^{-1}) \right), \quad \mbox{ as } n \to \infty.
\end{equation}
%\begin{equation}\label{Y11 at infinity}
%\phi_n(z) = e^{ng(z)} \frac{1}{2} (\beta(z) + \beta^{-1}(z)) \Theta_{11}(z) (1 + \bigO(n^{-1})),
%\end{equation} for $|z|>1+\epsilon$,
%and
%\begin{equation}\label{Y11 near 0}
%\phi_n(z) = e^{ng(z)}e^{n\pi i \Omega} \frac{1}{2} (\beta(z) + \beta^{-1}(z)) \Theta_{11}(z) (1 + \bigO(n^{-1})),
%\end{equation}
%for $|z|<1-\epsilon$.

For the norming constants $h_n$, we have $h_n=Y_{12}(0)$ by \eqref{sol_Y}.
Using the fact that $g(0) = \pi i$ and the identity $\beta(0) - \beta^{-1}(0) = -2i\sin \left( \frac{L+T}{2}\right)$, we obtain
\begin{align}
h_{n} &= Y_{12}(0) = e^{-n\ell} P_{12}^{(\infty)}(0) \left( 1 + \bigO(n^{-1}) \right) \nonumber\\
&= e^{-n\ell} \sin \left( \frac{L+T}{2} \right) \Theta_{12}(0)\left( 1 + \bigO(n^{-1}) \right),\label{norm of orthogonal pol}
\end{align}
as $n \to \infty$, where we note that $\Theta_{12}(0)$ is real since $u(0)-d \in \mathbb{R}$ and satisfies
\begin{equation}
\frac{\theta(1/2)}{\theta(0)} \leq \Theta_{12}(0) \leq \frac{\theta(0)^{2}}{\theta(1/2)^{2}}.
\end{equation}

\subsubsection{Inside the lenses away from edges}
In this subsection, we will only consider the cases where $z$ is outside the disks near the edge points. For $z$ on the unit circle, or close to it, we need to use the formula for $S$ valid inside the lens-shaped regions, see \eqref{def S}.  This leads to different asymptotic expressions for $\phi_n$.

Inside the lenses around $\gamma$, for $|z| > 1$, we have
\begin{equation}\label{poly lens}
\begin{array}{r c l}
\phi_{n}(z) & = & \displaystyle e^{ng(z)} \sum_{j=1,2}  R_{1j}(z) \left( P_{j1}^{(\infty)}(z) - e^{-n\phi(z)} e^{n\pi i \Omega} P_{j2}^{(\infty)}(z) \right) \\
& = & \displaystyle e^{ng(z)} \left( P_{11}^{(\infty)}(z) - e^{-n\phi(z)} e^{n\pi i \Omega} P_{12}^{(\infty)}(z) + \bigO(n^{-1})\right),
\end{array}
\end{equation}
as $n\to\infty$,
since $\Re (\phi(z)) > 0$ for $z$ inside the lenses and $P^{(\infty)} = \bigO(1)$ in this region.

Similarly to \eqref{poly lens}, we obtain asymptotics for $\phi_{n}$ in the other regions. For $|z| <1$, inside the lenses around $\gamma$, we have
\begin{equation}
\phi_{n}(z) = e^{ng(z)} \left( P_{11}^{(\infty)}(z) + e^{-n\phi(z)}e^{-n\pi i \Omega} P_{12}^{(\infty)}(z)+ \bigO(n^{-1}) \right) \quad \mbox{ as } n \to + \infty.
\end{equation}  

For $z$ inside the lenses around $\widetilde{\gamma}$, we obtain similar expressions:
\begin{equation}
\phi_{n}(z) = e^{ng(z)} \left( P_{11}^{(\infty)}(z) - e^{-n\widetilde\phi(z)}e^{-n\pi i \Omega} P_{12}^{(\infty)}(z) + \bigO(n^{-1})\right) \quad \mbox{ as } n \to + \infty
\end{equation}
for $|z| > 1$ and 
\begin{equation}
\phi_{n}(z) = e^{ng(z)} \left( P_{11}^{(\infty)}(z) + e^{-n\widetilde\phi(z)}e^{n\pi i \Omega} P_{12}^{(\infty)}(z) + \bigO(n^{-1})\right)\quad \mbox{ as } n \to + \infty
\end{equation}
for $|z| < 1$.

\subsubsection{Near the hard edges} \label{Near the hard edges}

Inside the disk $D(z_0,r)$, for $|z|>1$ and $z$ inside the lenses, inverting the transformations in the RH analysis, we obtain the identity
\begin{equation}
Y(z) = e^{\frac{n \pi i }{2} \sigma_{3}} e^{-\frac{n\ell}{2}\sigma_{3}} R(z) P(z) \begin{pmatrix}
1 & 0 \\ -e^{-n\phi(z)}e^{n\pi i \Omega} & 1
\end{pmatrix} e^{ng(z) \sigma_{3}} e^{\frac{n\ell}{2}\sigma_{3}} e^{-\frac{n\pi i }{2}\sigma_{3}}.
\end{equation}
Taking the $1,1$ entry, we have
\begin{multline}\label{poly near hard edge}
\phi_{n}(z) = e^{ng(z)} \left[ \left( P_{11}(z) - e^{-n\phi(z)}e^{n\pi i \Omega}P_{12}(z) \right)R_{11}(z)\right. \\\left. +\left( P_{21}(z) - e^{-n\phi(z)}e^{n\pi i \Omega}P_{22}(z) \right) R_{12}(z) \right],
\end{multline}
where $P$ is the local Bessel parametrix \eqref{P local para Bessel}. Similar expressions hold near $\overline{z_0}$; they can easily be obtained from the equation $\phi_n(\overline{z_0})=\overline{\phi_n(z_0)}$.

\subsubsection{Near the soft edges}

Near $\overline{z_1}$, we can express $\phi_n$ asymptotically in terms of the Airy parametrix. For $|z| > 1$, $z$ inside the lenses around $\widetilde{\gamma}$, and $z\in D(\overline{z_{1}},r)$, we obtain after a similar calculation than in Section \ref{Near the hard edges} that
\begin{multline}\label{poly soft edge}
\phi_{n}(z) = e^{ng(z)} \left[ \left( P_{11}(z) - e^{-n\widetilde\phi(z)}e^{-n\pi i \Omega}P_{12}(z) \right)R_{11}(z)\right. \\\left. +\left( P_{21}(z) - e^{-n\widetilde\phi(z)}e^{-n\pi i \Omega}P_{22}(z) \right) R_{12}(z) \right],
\end{multline}
where $P$ is the local Airy parametrix \eqref{def P Airy}.

\section{Proofs of main results}\label{section: proofs}

In this section, we use the asymptotics obtained for the orthogonal polynomials $\phi_n$ to prove our main results, Theorem \ref{theorem: Toeplitz}, Theorem \ref{theorem macro}, and Theorem \ref{theorem kernel}, corresponding to Case IV where $s=e^{-xn}$ with $x$ in compact subsets of $(0,x_{c})$. Afterwards we will also indicate how one can prove the corresponding results in Cases I-III, without giving full details.

\subsection{Zeros of $\phi_n$}
\label{section: proof zeros IV}
If we denote by $z_{j}^{(n)}$, $j = 1,...,n$ the zeros of the polynomial $\phi_{n}(z)$ and by $\nu_{n,s,L}=\frac{1}{n}\sum_{j=1}^n\delta_{z_j^{(n)}}$ the normalized zero counting measure of $\phi_n$, we have by definition that
\[
\phi_{n}(z) = \prod_{j=1}^{n} (z-z_{j}^{(n)}) = e^{\sum_{j=1}^{n} \log (z-z_{j}^{(n)})} = e^{n \int \log(z-w)d\nu_{n,s,L}(w)}.
\]
We know from the general theory for orthogonal polynomials on the unit circle that all the zeros lie strictly inside the unit circle. It follows from Helly's theorem that there exists a subsequence $(\nu_{n_{k},s,L})_{k}$ which converges weakly to a probability measure $\nu_{s,L}$ supported on a subset of $|z| \leq 1$.
From \eqref{Y11 near inf and 0}, it follows that
\[
\lim_{n\to \infty} \frac{1}{n} \mathrm{Re}(\log \phi_{n}(z))  = \int \log |z-e^{i\theta}| d\mu_{x,L}(e^{i\theta}), \qquad \forall z \in  \mathbb{C}\setminus S^1.
\]
In particular we have for any converging subsequence $\nu_{n_k,s,L}$ that
\[
\lim_{k\to \infty} \int\log |z-w|d\nu_{n_{k},s,L}(w) = \int \log |z-e^{i\theta}|d\mu_{x,L}(e^{i\theta}),
\]
for all $z$ off the unit circle.
By the unicity theorem \cite[Theorem II.2.1]{SaffTotik}, this implies that $\nu_{n_k, s,L}$ converges weakly to $\mu_{x,L}$. Hence the entire sequence
$\nu_{n,s,L}$ converges weakly to $\mu_{x,L}$.
This proves a first part of Theorem \ref{theorem macro}.

%We can even prove slightly more than weak convergence of the zero counting measures: for $n$ large enough, there can at most be one zero which is not close to the unit circle, as stated in the following proposition.
We can also prove a result about the location of the zeros of $\phi_{n}$: for $n$ large enough, there cannot be zeros which are not close to the unit circle or to a point $w_{n}$ on the real line, as stated in the following proposition.
\begin{proposition} \label{prop: zeros phin}
Let $0<L<\pi$ and $s=e^{-xn}$ with $0<x < x_{c}$. There exists a real sequence $(w_{n})_{n\in\mathbb N}$ such that for any $\epsilon>0$, there exists $n_0\in\mathbb N$ such that for every $n\geq n_0$, $\phi_{n}(z)$ has no zeros in $\{z\in\mathbb C: |z| \leq 1-\epsilon, |z-w_{n}| \geq \epsilon\}$.
\end{proposition}
\begin{proof}
We can use \eqref{Y11 near inf and 0} for $|z|\leq 1-\epsilon$ and $n$ sufficiently large, which implies that $\phi_n(z)$ has, for sufficiently large $n$, no zeros if $|P^{(\infty)}_{11}(z)|>c>0$ with $c$ independent of $n$. Since $\beta(z)+\beta^{-1}(z)$ and $\theta(u(z)+d)$ both are bounded away from zero in this region, $P_{11}^{(\infty)}(z)$ vanishes only at the zeros of $\theta(u(z)+d-n\Omega)$. 
Since the Riemann surface $X$ is of genus one, $u(z)$ is a bijection from $X$ to the Jacobi variety $\mathbb{C}/ \Lambda$, $\Lambda = \{ n + \tau m , n,m \in \mathbb{Z}\}$, implying that $\theta(u(z)+d-n\Omega)$ vanishes in at most one point which we denote by $w_{n}$. For any $n$, the point $w_n$ is necessarily real because of the symmetry $P^{(\infty)}(z) = \overline{P^{(\infty)}(\overline{z})}$. This relation follows from the fact that $\overline{P^{(\infty)}(\overline{z})}$ satisfies the RH conditions for $P^{(\infty)}$, and the uniqueness of the RH solution.
\end{proof}

%\begin{remark}
%We can say a bit more about $w_n$ than we did in the proof.
%Note that if $n\Omega = 0 \mod 1$, $P_{11}^{(\infty)}(z)$ vanishes only at $z_{\star}^{(2)}$. Since $n \Omega \in \mathbb{R}$, from the periodicity properties of theta function \eqref{periodicity of theta}, we can assume that $n \Omega \in [0,1)$. On the other hand from a residue calculation we easily obtain
%\begin{equation*}
%\int_{\infty^{(2)}}^{-1_{-}^{(2)}} \omega + \int_{1_{-}^{(2)}}^{\infty^{(2)}} \omega = - \frac{1}{2},
%\end{equation*}
%and by a change of variable
%\begin{equation*}
%\int_{-1_{+}^{(1)}}^{1_{+}^{(1)}} \omega = - \frac{1}{2}.
%\end{equation*}
%Therefore $P_{11}^{(\infty)}(z)$ vanishes in only one point $w_{n}$ on the Riemann surface $X$ which is located on 
%\begin{equation*}
%\{z \in \mathbb{R} : |z|\geq 1, z\mbox{ on the second sheet}\} \cup \{ z \in \mathbb{R} : |z|\leq 1, z\mbox{ on the first sheet}\},
%\end{equation*}
%and the result follows.
%\end{remark}

\subsection{Limit of the mean eigenvalue density}

We first prove that the mean eigenvalue density $\psi_{n,s,L}(e^{i\theta})=\frac{1}{n}K_n(e^{i\theta},e^{i\theta})$ converges to $0$ for $e^{i\theta}$ on $S^1\setminus\supp\,\mu_{x,L}$.

For $z$ in the exterior region of $\Sigma_{R}$, we have
\eqref{Y11 near inf and 0}. If we let $z$ approach a point $e^{i\theta}\in S^1\setminus\supp\,\mu_{x,L} = \Sigma_1\cup\Sigma_2$, we obtain
\[
\phi_{n}(e^{i\theta}) = e^{ng_{-}(e^{i \theta})}P_{11,-}^{(\infty)}(e^{i\theta}) \left( 1 + \bigO(n^{-1}) \right),\qquad n\to\infty,
\]
where we note that $P^{(\infty)}_{11,-}(e^{i\theta})$ is bounded away from $0$.
A direct calculation shows that 
\[
\phi_n'(e^{i\theta}) \overline{\phi_n(e^{i\theta})} - \left( \phi_n^* \right)' (e^{i\theta}) \overline{\phi_n^*(e^{i\theta})} = e^{2n \mathrm{Re}(g_{-}(e^{i\theta}))}\bigO(n),\qquad n\to\infty.
\]
It was shown in \cite[Equations (3.16) and (3.21)]{ChCl}  that
\[
2 \mbox{Re}(g_{-}(e^{i\theta})) = 2 \int_{-\pi}^{\pi} \log |e^{i\theta}-e^{i\alpha}|d\mu(e^{i\alpha}) < -\ell+x 
\]
for $e^{i\theta}$ in $\Sigma_1\cup\Sigma_2$.
Combining this with \eqref{onepointfunction} and \eqref{norm of orthogonal pol}, we obtain
\begin{equation}
\psi_{n,s,L}(e^{i\theta}) = \exp \left( \left( 2 \int_{-\pi}^{\pi} \log |e^{i\theta}-e^{i\alpha}|d\mu(e^{i\alpha}) + \ell - x  \right)n \right) \bigO(1),\qquad n\to\infty.
\end{equation}
This shows that $\psi_{n,s,L}(e^{i\theta})$ tends to $0$ uniformly for $e^{i\theta}$ in any compact subset of $\Sigma_1\cup\Sigma_2$.

\medskip

Next, we show that $\psi_{n,s,L}(e^{i\theta})\to\psi_{x,L}(e^{i\theta})$ uniformly for $e^{i\theta}$ in compact subsets of $\gamma$, where $\psi_{x,L}$ is given by \eqref{def muxL}.
Taking the limit as $z\to e^{i\theta}$ in \eqref{poly lens}, we get
\begin{equation} \label{pol ortho near gamma}
\phi_{n}(e^{i\theta}) = e^{n g_{-}(e^{i\theta})} \left( Q(e^{i\theta}) + \bigO(n^{-1}) \right),\qquad n\to\infty,
\end{equation}
where $Q(e^{i\theta}) = P_{11,-}^{(\infty)}(e^{i\theta}) - e^{-n\phi_{-}(e^{i\theta})}e^{n\pi i \Omega} P_{12,-}^{(\infty)}(e^{i\theta}) $.

Now, \eqref{onepointfunction} can be rewritten as 
\begin{equation} \label{onepointfunction 2}
\psi_{n,s,L}(e^{i\theta}) = \frac{f(e^{i\theta})}{2\pi n h_{n}} \left[ -n |\phi_{n}(e^{i\theta})|^{2} + 2\mbox{Re} \left( e^{i\theta}\phi_{n}^{\prime}(e^{i\theta}) \overline{\phi_{n}(e^{i\theta})} \right) \right].
\end{equation}

Substituting \eqref{norm of orthogonal pol} and \eqref{pol ortho near gamma} in \eqref{onepointfunction 2}, taking into account that $\mbox{Re}(g_{-}(e^{i\theta})) = -\frac{\ell}{2}$ and $f(e^{i\theta}) = 1$ for $e^{i\theta} \in \gamma$, we obtain
\begin{multline*}
\psi_{n,s,L}(e^{i\theta}) = \frac{1}{2\pi P_{12}^{(\infty)}(0)} \left[ |Q(e^{i\theta})|^{2} \left( 2 \mbox{Re} (e^{i\theta} g_{-}^{\prime}(e^{i\theta})) - 1 \right) \right. \\ + \left. 2 \mathrm{Re} \left( e^{i\theta} \phi_{-}^{\prime}(e^{i\theta}) e^{-n \phi_{-}(e^{i\theta})} e^{n\pi i \Omega} P_{12,-}^{(\infty)}(e^{i\theta}) \overline{Q(e^{i\theta})} \right) +\bigO(n^{-1}) \right],
\end{multline*}
as $n\to\infty$.

Equation \eqref{def_of_phi} together with
the properties 
\begin{equation} \label{g properties on gamma}
\begin{array}{r c l}
\displaystyle g_{+}(z) + g_{-}(z) & = & \displaystyle \log z + i\pi - \ell, \qquad \hspace{4.4 mm} \mbox{ for } z \in \gamma, \\
\displaystyle g_{+}(z) - g_{-}(z) & = & \displaystyle 2 \pi i \int_{\arg z}^{\pi} d\mu(e^{i\theta}), \qquad  \mbox{ for } z \in S^{1}, \\
\end{array}
\end{equation}
where $-\pi < \arg z \leq \pi$, implies
\[
2 \mbox{Re} (e^{i\theta} g_{-}^{\prime}(e^{i\theta})) - 1 = e^{i\theta} \phi_{-}^{\prime}(e^{i\theta}) = 2\pi \psi_{x,L}(e^{i\theta}).
\]
We thus have, as $n\to\infty$,
\begin{equation} \label{modulusform}
\psi_{n,s,L}(e^{i\theta}) = \psi_{x,L}(e^{i\theta}) \left( \frac{|P_{11,-}^{(\infty)}(e^{i\theta})|^{2}-|P_{12,-}^{(\infty)}(e^{i\theta})|^{2}}{P_{12}^{(\infty)}(0)}  + \bigO(n^{-1}) \right).
\end{equation}
The convergence of $\psi_{n,s,L}$ to $\psi_{x,L}$ then follows from the following identity.
\begin{proposition}\label{prop Pinf relation}
For all $z \in \gamma \cup \widetilde\gamma$, we have
\begin{equation}\label{eq Pinf relation}
|P_{11,-}^{(\infty)}(e^{i\theta})|^{2}-|P_{12,-}^{(\infty)}(e^{i\theta})|^{2} = P_{12}^{(\infty)}(0).
\end{equation}
\end{proposition}
\begin{proof}
We recall the relation $P^{(\infty)}(z) = \overline{P^{(\infty)}(\overline{z})}$ (see the proof of Proposition \ref{prop: zeros phin}). This allows us to rewrite the left hand side of \eqref{eq Pinf relation} as
\begin{equation} \label{nonmodulusform}
P_{11,-}^{(\infty)}(e^{i\theta}) P_{11,-}^{(\infty)}(e^{-i\theta})- P_{12,-}^{(\infty)}(e^{i\theta}) P_{12,-}^{(\infty)}(e^{-i\theta}).
\end{equation}
By another argument of uniqueness of the solution of the RH problem for $P^{(\infty)}$, we have the relation
\begin{equation} \label{inverse_relation_of_Pinf}
P^{(\infty)}(z;\Omega) = \sigma_{1} P^{(\infty)}(0;-\Omega)^{-1} P^{(\infty)}(z^{-1};-\Omega) \sigma_{1},
\end{equation}
where we have now explicitly written the dependence on the parameter $\Omega$ of $P^{(\infty)}$, and $\sigma_{1} = \begin{pmatrix}
0 & 1 \\ 1 & 0
\end{pmatrix}$. Noting from \eqref{final expression Pinf} that $P_{21}^{(\infty)}(0;-\Omega) = -P_{12}^{(\infty)}(0;\Omega)$ and using \eqref{inverse_relation_of_Pinf}, we obtain
\[
P_{12}^{(\infty)}(0;\Omega) = P_{22}^{(\infty)}(z^{-1};-\Omega) P_{12}^{(\infty)}(z;\Omega) - P_{21}^{(\infty)}(z^{-1};-\Omega) P_{11}^{(\infty)}(z;\Omega).
\]
From \eqref{final expression Pinf}, for $|z| \neq 1$, 
\begin{equation}
P_{22}^{(\infty)}(z;-\Omega) = P_{11}^{(\infty)}(z;\Omega) \qquad \mbox{ and } \qquad P_{21}^{(\infty)}(z;-\Omega) = -P_{12}^{(\infty)}(z;\Omega).
\end{equation}
We are then led to an expression involving only $P^{(\infty)}(z;\Omega)$ (so we omit again the dependence in $\Omega$):
\[
P_{12}^{(\infty)}(0) = P_{11}^{(\infty)}(z^{-1})P_{12}^{(\infty)}(z) + P_{12}^{(\infty)}(z^{-1})P_{11}^{(\infty)}(z), \quad \mbox{ for } |z| \neq 1.
\]
Taking the limit of the above expression as $z \to e^{i\theta} \in \gamma \cup \widetilde \gamma$, $|z| > 1$ and noting that
\begin{align*}
&\lim_{z\to e^{i\theta}_-}P_{11}^{(\infty)}(z^{-1})  =  P_{11,+}^{(\infty)}(e^{-i\theta}) = -P_{12,-}^{(\infty)}(e^{-i\theta}), \\
&\lim_{z\to e^{i\theta}_-}P_{12}^{(\infty)}(z^{-1})_{-}  =  P_{12,+}^{(\infty)}(e^{-i\theta}) = P_{11,-}^{(\infty)}(e^{-i\theta}), \end{align*}
we obtain that $P_{12}^{(\infty)}(0)$ is equal to \eqref{nonmodulusform}, which proves the proposition.
\end{proof}
The convergence of $\psi_{n,s,L}(e^{i\theta})$ to $\psi_{x,L}(e^{i\theta})$ for $e^{i\theta}$ on $\widetilde\gamma$ is obtained in a similar way.

To complete the proof of Part 1 of Theorem \ref{theorem macro}, we still need to show weak convergence of $\mu_{n,s,L}$ to $\mu_{x,L}$. 
To that end, let $h$ be a continuous function on the unit circle.
For any $\epsilon>0$, we can take a compact subset $A$ of the unit circle which does not contain the points $z_0, \overline{z_0}, z_1, \overline{z_1}$ and which is such that 
\[\mu_{x,L}(A)>1-\frac{\epsilon}{3\max_{z\in S^1}|h(z)|}.\]
Because $\psi_{n,s,L}$ converges uniformly to $\psi_{x,L}$ for $e^{i\theta}\in A$, we then have
\begin{equation}\label{estimate weakconv1}\left|\int_{e^{i\theta}\in A}h(e^{i\theta})\left(\psi_{n,s,L}(e^{i\theta})-\psi_{x,L}(e^{i\theta})\right)d\theta\right|<\epsilon/3,\end{equation}
for $n$ sufficiently large.
We also have
\begin{equation}\label{estimate weakconv2}\left|\int_{e^{i\theta}\in S^1\setminus A}h(e^{i\theta})\psi_{x,L}(e^{i\theta})d\theta \right|\leq \max_{z\in S^1}|h(z)|(1-\mu_{x,L}(A))<\epsilon/3.\end{equation}
Since $\mu_{n,s,L}(A)$ converges to $\mu_{x,L}(A)$ as $n \to\infty$, we also have
\begin{equation}\label{estimate weakconv3}\left|\int_{e^{i\theta}\in S^1\setminus A}h(e^{i\theta})\psi_{n,s,L}(e^{i\theta})d\theta \right|\leq \max_{z\in S^1}|h(z)|(1-\mu_{n,s,L}(A))<\epsilon/3,\end{equation}
for sufficiently large $n$. Combining \eqref{estimate weakconv1}--\eqref{estimate weakconv3}, we obtain
\begin{equation}\label{estimate weakconv4}\left|\int_{0}^{2\pi} h(e^{i\theta})\left(\psi_{n,s,L}(e^{i\theta})-\psi_{x,L}(e^{i\theta})\right)d\theta\right|<\epsilon,\end{equation}
which proves the weak convergence of $\mu_{n,s,L}$ to $\mu_{x,L}$.

The asymptotics \eqref{limOmega} for the average number of eigenvalues are obtained in a similar way: we have
\[\mathbb E\left(\#(\Theta\cap \gamma^c)|\Phi\subset\gamma\right)=n\int_{e^{i\theta}\in\gamma^c}\psi_{n,s,L}(e^{i\theta})d\theta.\] Repeating the same argument as above, but now with $h$ the characteristic function of $\gamma^c$, we obtain \eqref{limOmega}. 

For the variance, we use \eqref{identity variance}, and obtain \eqref{thm variance} after differentiating \eqref{limOmega} with respect to $s=e^{-xn}$. This can be justified using a standard argument in RH analysis: first, it follows from the RH analysis that \eqref{limOmega} holds not only for real $0<x<x_c$, but also uniformly for $x$ in a small complex neighbourhood of $(0,x_c)$. Next, using Cauchy's integral formula 
\[\partial_x \left(\partial_x \log D_n(e^{-nx},L)\right)=\frac{1}{2\pi i}\int \frac{\partial_{x'} \log D_n(e^{-nx'},L)}{(x-x')^2}dx',\] where the integral is over a small circle in the $x$-plane. Using \eqref{limOmega} and the fact that the error term is uniform on the little circle, we obtain \eqref{thm variance}.

\subsection{Proof of Theorem \ref{theorem: Toeplitz}}

Since \[\frac{1}{n}\left. s\partial_s\log D_n(s,L)\right|_{s=e^{-xn}}=\frac{1}{n}\mathbb E\left(\#(\Theta\cap \gamma^c)|\Phi\subset\gamma\right)\leq 1,\] 
we can integrate \eqref{limOmega} with respect to $x$ and use Lebesgue's dominated convergence theorem. From the point-wise convergence \eqref{limOmega} for $0<x<x_c$, we obtain
\begin{align*}\label{diff id as}
\lim_{n\to\infty} \frac{1}{n^2} \int_0^x \partial_{\xi}\log D_n(e^{-\xi n},L) d\xi&=-\lim_{n\to\infty} \frac{1}{n}
\int_{0}^{x} \left. s \partial_s \log D_n(s,L) \right|_{s=e^{-\xi n}}d\xi \\&=-\int_{0}^{x} \Omega(\xi,L) d\xi,
\end{align*}
for $x\in[0,x_c]$.

\subsection{Bessel kernel limit}\label{section: proof Bessel}

Let $u$, $v$ be two positive real numbers and $c$ a positive constant. To simplify the expressions we will use the following notations:
\begin{equation}
u_{n} = L - \frac{u}{(cn)^{2}}, \qquad v_{n} = L - \frac{v}{(cn)^{2}}.
\end{equation}
Using \eqref{norm of orthogonal pol} in \eqref{kernel conditional CUE CD}, we obtain as $n\to\infty$,
\begin{multline}\label{Bessel kernel 1}
\frac{1}{(cn)^{2}}K_{n}(e^{iu_{n}},e^{iv_{n}}) = \frac{-i\left( e^{in(u_{n}-v_{n})} \phi_{n}(e^{-iu_{n}}) \phi_{n}(e^{iv_{n}}) - \phi_{n}(e^{iu_{n}}) \phi_{n}(e^{-iv_{n}})) \right)}{2\pi e^{-n\ell} P_{12}^{(\infty)}(0)(u-v)}\\
\times\  (1 + \bigO(n^{-1})).
\end{multline}
We can also use the asymptotic expression \eqref{poly near hard edge} for $\phi_{n}(e^{iu_{n}})$.
Noting that
\[
n^{2} \zeta(e^{i u_{n}}) = -n^{2} \zeta^{\prime}(z_{0}) e^{iL} \frac{iu}{(cn)^{2}} \left( 1 + \bigO\left(\frac{u}{n^{2}}\right) \right),
\] 
and that
\[\zeta^{\prime}(z_{0}) = \frac{1}{4} e^{-iL} e^{-i\frac{\pi}{2}} \frac{|z_{0}-z_{1}||z_{0}-\overline{z_{1}}|}{|z_{0}-\overline{z_{0}}|},
\]
we obtain
\begin{equation} \label{ncarre zeta}
n^{2}\zeta(e^{i u_{n}}) = -\frac{u}{4}\left( 1 + \bigO\left(\frac{u}{n^{2}}\right) \right),\qquad n^{2}\zeta(e^{i v_{n}}) = -\frac{v}{4}\left( 1 + \bigO\left(\frac{v}{n^{2}}\right) \right),
\end{equation} 
as $n\to\infty$,
if we choose $c = \sqrt{\frac{|z_{0}-z_{1}||z_{0}-\overline{z_{1}}|}{|z_{0}-\overline{z_{0}}|}}$.
This implies by \eqref{P local para Bessel} and \eqref{E in Bessel} that $P_{-}(e^{iu_{n}}) = \bigO(n^{1/2})$ as $n\to\infty$ uniformly for $u$ in a compact subset of $(0,\infty)$, and similarly for $v_n$. By \eqref{as R} and \eqref{poly near hard edge}, we thus have
\begin{equation} \label{pol Bessel 1}
\phi_{n}(e^{iu_{n}}) = e^{ng_{-}(e^{iu_{n}})} \left( P_{11,-}(e^{iu_{n}}) - e^{-n\phi_{-}(e^{iu_{n}})}e^{n\pi i \Omega}P_{12,-}(e^{iu_{n}})+\bigO(n^{-1/2}) \right),
\end{equation}
as $n\to\infty$.

The definition of $\widehat{\Psi}$ and one of the connection formulas for the Bessel functions (see \cite{NIST}),
\[
H_{0}^{(1)}(z) + H_{0}^{(2)}(z) = 2 J_{0}(z),
\]
imply by \eqref{P local para Bessel} that, for $\delta<x<x_c-\delta$,
\[
P_{11,-}(e^{iu_{n}}) - e^{-n\phi_{-}(e^{iu_{n}})}e^{n\pi i \Omega}P_{12,-}(e^{iu_{n}}) = e^{-\frac{n}{2}\phi_{-}(e^{iu_{n}})} e^{n\pi i \frac{\Omega}{2}}  F(e^{iu_{n}})  + \bigO(e^{-nx}),
\]
as $n\to\infty$,
where
\[
F(z) = E_{11}(z) J_{0}(2\sqrt{-n^{2}\zeta(z)}) + 2i \pi \sqrt{-n^{2}\zeta(z)} E_{12}(z) J_{0}^{\prime}(2\sqrt{-n^{2}\zeta(z)}).
\]

Using \eqref{def_of_phi} and \eqref{g properties on gamma}, we also have the following relation,
\begin{equation}\label{g phi relation in kernel Bessel}
g_{-}(e^{i\theta}) - \frac{1}{2} \phi_{-}(e^{i\theta}) + \pi i \frac{\Omega}{2} = -\frac{\ell}{2} + \frac{i(\theta + \pi)}{2}, \quad \mbox{ for } -L \leq \theta \leq L,
\end{equation}
which allows us to rewrite \eqref{pol Bessel 1} in a more explicit form in terms of the Bessel function of the first kind:
\begin{equation}\label{phi u_n}
\phi_{n}(e^{iu_{n}}) = e^{-\frac{n\ell}{2}} e^{\frac{ni}{2}(u_{n}+\pi)} \left( F(e^{iu_{n}}) +\bigO(n^{-1/2})\right),\qquad n\to\infty.
\end{equation}

Substituting \eqref{phi u_n} in \eqref{Bessel kernel 1}, we get
\begin{equation}\label{Bessel kernel in terms of F}
\frac{1}{(cn)^{2}} K_{n}(e^{iu_{n}},e^{iv_{n}}) = \frac{-1}{\pi P_{12}^{(\infty)}(0)(u-v)} \mbox{Im} \left( F(e^{iu_{n}})\overline{F(e^{iv_{n}})} \right) +\bigO(n^{-1/2}),
\end{equation}
as $n\to\infty$.
We now compute $F(e^{iu_{n}})\overline{F(e^{iv_{n}})}$: we have
\[
\begin{array}{r c l}
F(e^{iu_{n}})\overline{F(e^{iv_{n}})} & = & F_{1} + F_{2} + F_{3} + F_{4}, \\
\end{array}
\]
where
\begin{equation*}
\begin{array}{r c l}
F_{1} & = & E_{11}(e^{iu_{n}})\overline{E_{11}(e^{iv_{n}})} J_{0}(2 \sqrt{-n^{2} \zeta(e^{iu_{n}})})J_{0}(2 \sqrt{-n^{2} \zeta(e^{iv_{n}})}), \\[0.15cm]
F_{2} & = & 4\pi^{2}n^2 \sqrt{\zeta(e^{iu_{n}})\zeta(e^{iv_{n}})} E_{12}(e^{iu_{n}})\overline{E_{12}(e^{iv_{n}})}J_{0}^{\prime}(2 \sqrt{-n^{2} \zeta(e^{iu_{n}})})J_{0}^{\prime}(2 \sqrt{-n^{2} \zeta(e^{iv_{n}})}), \\[0.15cm]
F_{3} & = & -2i\pi \sqrt{-n^{2} \zeta(e^{iv_{n}})} E_{11}(e^{iu_{n}})\overline{E_{12}(e^{iv_{n}})}J_{0}(2 \sqrt{-n^{2} \zeta(e^{iu_{n}})})J_{0}^{\prime}(2 \sqrt{-n^{2} \zeta(e^{iv_{n}})}), \\[0.15cm]
F_{4} & = & 2i\pi \sqrt{-n^{2} \zeta(e^{iu_{n}})} E_{12}(e^{iu_{n}})\overline{E_{11}(e^{iv_{n}})}J_{0}(2 \sqrt{-n^{2} \zeta(e^{iv_{n}})})J_{0}^{\prime}(2 \sqrt{-n^{2} \zeta(e^{iu_{n}})}).
\end{array}
\end{equation*}
Here, we have used the fact that $J_{0}(x)$ is real for $x > 0$. It is straightforward to see that the imaginary parts of $F_{1}$ and $F_{2}$ tend to $0$ as $n\to + \infty$. More precisely, one has
\begin{equation} \label{F_1 and F_2}
\mbox{Im}(F_{j}) = \bigO \left( \frac{|u| + |v|}{n} \right),\qquad j=1,2,\qquad n\to\infty.
\end{equation}
The computation for $F_{3}$ and $F_{4}$ is slightly more involved. Using \eqref{E in Bessel} and Proposition \ref{prop Pinf relation}, we have
\begin{equation} \label{F_3}
\mbox{Im}(F_{3}) = -\frac{\pi}{2} J_{0}(\sqrt{u})\sqrt{v}J_{0}^{\prime}(\sqrt{v}) P_{12}^{(\infty)}(0) \left( 1 + \bigO(n^{-1}) \right)
\end{equation}
and
\begin{equation} \label{F_4}
\mbox{Im}(F_{4}) = \frac{\pi}{2} J_{0}(\sqrt{v})\sqrt{u}J_{0}^{\prime}(\sqrt{u}) P_{12}^{(\infty)}(0) \left( 1 + \bigO(n^{-1}) \right),
\end{equation}
as $n\to\infty$. Note that it follows from Proposition \ref{prop Pinf relation} that 
$P_{12}^{(\infty)}(0)$ is real.
Putting everything together in \eqref{Bessel kernel in terms of F}, one gets
\begin{equation}
\frac{1}{(cn)^{2}}K_{n}(e^{iu_{n}},e^{iv_{n}}) = \frac{J_{0}(\sqrt{u}) \sqrt{v} J_{0}^{\prime}(\sqrt{v}) - J_{0}(\sqrt{v}) \sqrt{u} J_{0}^{\prime}(\sqrt{u})}{2(u-v)}+\bigO(n^{-1}),
\end{equation}
as $n\to\infty$,
uniformly for $u$, $v$ in compact subsets of $(0,\infty)$.

\subsection{Airy kernel  limit}\label{section: Airy}

For convenience, we make use of the following notations:
\[
u_{n} = -\pi + T + \frac{u}{(cn)^{2/3}}, \qquad v_{n} = -\pi + T + \frac{v}{(cn)^{2/3}}.
\]
Even if \eqref{Airy kernel limit} is true for any $u$, $v$ real, we will only provide explicit computation for $u,v < 0$. Other cases can be dealt with in a similar manner (and are in fact easier). Using \eqref{norm of orthogonal pol} in \eqref{kernel conditional CUE CD}, we obtain
\begin{multline} \label{Kernel u_n for Airy}
\frac{1}{(cn)^{2/3}}K_{n}(e^{iu_{n}},e^{iv_{n}}) = \frac{ie^{-xn}(e^{in(u_{n}-v_{n})}\phi_{n}(e^{-iu_{n}})\phi_{n}(e^{iv_{n}})-\phi_{n}(e^{iu_{n}})\phi_{n}(e^{-iv_{n}}))}{2\pi e^{-n\ell}P_{12}^{(\infty)}(0) (u-v)} \\ \times\ (1+\bigO(n^{-2/3})),\qquad n\to\infty.
\end{multline}

If we choose $c = \sqrt{\frac{|\overline{z_{1}}-z_{1}|}{4|\overline{z_{1}}-z_{0}||\overline{z_{1}}-\overline{z_{0}}|}}$, we have $ n^{2/3} \zeta(e^{iu_{n}}) = u(1+\bigO(\frac{u}{(cn)^{2/3}}))$ as $n\to\infty$. With this choice of $c$, it is straightforward to see from \eqref{def P Airy} and \eqref{E in airy} that $P_{-}(e^{iu_{n}}) = \bigO(n^{1/6})$. Taking the limit as $z \to e^{iu_{n}}$ in \eqref{poly soft edge}, we obtain the large $n$ asymptotics
\begin{align} 
\phi_{n}(e^{iu_{n}}) & =  e^{ng_{-}(e^{iu_{n}})}\left( P_{11,-}(e^{iu_{n}}) - e^{-n\widetilde{\phi}_{-}(e^{iu_{n}})}e^{-n\pi i \Omega}P_{12,-}(e^{iu_{n}}) +\bigO(n^{-5/6})\right), \nonumber\\
& = e^{n(g_{-}(e^{iu_{n}})-\frac{1}{2}\widetilde{\phi}_{-}(e^{iu_{n}}) - \frac{1}{2}\pi i \Omega)}\left(K(e^{iu_{n}})+\bigO(n^{-5/6})\right), \nonumber\\
& = e^{n \frac{-\ell+x}{2}}e^{n i \frac{u_{n}-\pi}{2}} \left(K(e^{iu_{n}})+\bigO(n^{-5/6})\right),
\label{phi_n u_n for Airy}
\end{align}
where $K$ is defined by
\begin{equation*}
K(z) = E_{11}(z) \mathrm{Ai}(n^{2/3}\tilde{\zeta}(z)) + E_{12}(z) \mathrm{Ai}^{\prime}(n^{2/3}\tilde{\zeta}(z)).
\end{equation*}
Note that we have used the analogue of \eqref{g phi relation in kernel Bessel} for $e^{i\theta} \in \widetilde \gamma$ with $-\pi < \theta \leq -\pi + T$, namely
\begin{equation}\label{g phi relation in kernel Airy}
g_{-}(e^{i\theta}) - \frac{1}{2} \widetilde\phi_{-}(e^{i\theta}) - \pi i \frac{\Omega}{2} = -\frac{\ell}{2} + \frac{x}{2} + \frac{i(\theta - \pi)}{2}, \quad -\pi < \theta \leq -\pi + T,
\end{equation}
and the Airy function relation $y_{0}(z) + y_{1}(z) + y_{2}(z) = 0$ (see \cite{NIST}) which implies by \eqref{sol Airy} that, for $-\pi < \arg z < - \frac{2\pi}{3}$, 
\[
\begin{array}{r c l}
\Phi_{11}(z) - \Phi_{12}(z) & = & \mbox{Ai}(z), \\[0.15cm]
\Phi_{21}(z) - \Phi_{22}(z) & = & \mbox{Ai}^{\prime}(z).
\end{array}
\]

Inserting \eqref{phi_n u_n for Airy} in \eqref{Kernel u_n for Airy}, we have as $n\to\infty$,
\[
\frac{1}{(cn)^{2/3}}K_{n}(e^{iu_{n}},e^{iv_{n}}) = \frac{e^{i \frac{n}{2} (u_{n}-v_{n}) }}{\pi P_{12}^{(\infty)}(0)(u-v)} \mbox{Im} (K(e^{iu_{n}})\overline{K(e^{iv_{n}})})+\bigO(n^{-2/3}).
\]
Using the fact that $\mbox{Ai}(x) \in \mathbb{R}$ for $x$ real, we have
\begin{equation}
K(e^{iu_{n}})\overline{K(e^{iv_{n}})} = K_{1} + K_{2} + K_{3} + K_{4},
\end{equation}
with
\[
\begin{array}{r c l}
K_{1} & = & E_{11}(e^{iu_{n}})\overline{E_{11}(e^{iv_{n}})} \mbox{Ai}(n^{2/3}\widetilde\zeta(e^{iu_{n}}))\mbox{Ai}(n^{2/3}\widetilde\zeta(e^{iv_{n}})), \\[0.15cm]
K_{2} & = & E_{12}(e^{iu_{n}})\overline{E_{12}(e^{iv_{n}})} \mbox{Ai}^{\prime}(n^{2/3}\widetilde\zeta(e^{iu_{n}}))\mbox{Ai}^{\prime}(n^{2/3}\widetilde\zeta(e^{iv_{n}})), \\[0.15cm]
K_{3} & = & E_{11}(e^{iu_{n}})\overline{E_{12}(e^{iv_{n}})} \mbox{Ai}(n^{2/3}\widetilde\zeta(e^{iu_{n}}))\mbox{Ai}^{\prime}(n^{2/3}\widetilde\zeta(e^{iv_{n}})), \\[0.15cm]
K_{4} & = & E_{12}(e^{iu_{n}})\overline{E_{11}(e^{iv_{n}})} \mbox{Ai}^{\prime}(n^{2/3}\widetilde\zeta(e^{iu_{n}}))\mbox{Ai}(n^{2/3}\widetilde\zeta(e^{iv_{n}})).\\
\end{array}
\]
Analogously to \eqref{F_1 and F_2}, \eqref{F_3} and \eqref{F_4}, we have
\begin{equation}
\mbox{Im}(K_{j}) = \bigO\left(\frac{|u| + |v|}{n^{1/3}}\right) ,\qquad j=1,2,\qquad n\to\infty,
\end{equation}
and
\[
\mbox{Im}(K(e^{iu_{n}})\overline{K(e^{iv_{n}})}) = (\mbox{Ai}(u)\mbox{Ai}^{\prime}(v) - \mbox{Ai}^{\prime}(u)\mbox{Ai}(v)) \mbox{Im}(E_{11}(\overline{z_{1}})\overline{E_{12}(\overline{z_{1}})})+ \bigO((u-v)n^{-1/3}),
\]
as $n\to\infty$.
Using the explicit expression for $E$ given by \eqref{E in airy} and Proposition \ref{prop Pinf relation}, one gets
\begin{equation}
\mbox{Im}(E_{11}(\overline{z_{1}})\overline{E_{12}(\overline{z_{1}})}) = \pi P_{12}^{(\infty)}(0),
\end{equation}
and we obtain the Airy kernel limit \eqref{Airy kernel limit}.

\subsection{Sine kernel limit}\label{section: sine}

For the sine kernel, we use the notations
\[
u_{n} = w + \frac{u}{cn} \quad \mbox{ and } \quad v_{n} = w + \frac{v}{cn},
\]
and in this section we will complete the computation only for $w \in (-L,L)$; the case $w \in (\pi - T,\pi+T)$ is similar. In this case, the expression for the kernel takes  the asymptotic form as $n\to\infty$,
\begin{multline}\label{Kernel in sine case}
\frac{1}{cn} K_{n}(e^{iu_{n}},e^{iv_{n}}) = \frac{i(e^{in(u_{n}-v_{n})}\phi_{n}(e^{-iu_{n}})\phi_{n}(e^{iv_{n}})-\phi_{n}(e^{iu_{n}})\phi_{n}(e^{-iv_{n}}))}{2\pi e^{-n\ell}P_{12}^{(\infty)}(0)(u-v)}\\\times\ (1+\bigO(n^{-1})).
\end{multline}
Starting from \eqref{poly lens} and taking the limit as $z \to e^{i\theta} \in \gamma$, one has
\begin{equation}\label{phi for sine kernel}
\phi_{n}(e^{i\theta}) = e^{n g_{-}(e^{i\theta})} \left( \widetilde K(e^{i\theta}) +\bigO(n^{-1})  \right),
\end{equation}
where $\widetilde K(e^{i\theta}) = P_{11,-}^{(\infty)}(e^{i\theta}) - e^{-n\phi_{-}(e^{i\theta})}e^{n\pi i \Omega} P_{12,-}^{(\infty)}(e^{i\theta})$. Using \eqref{g properties on gamma} we immediately obtain for $e^{i\theta} \in \gamma$ that
\begin{equation}\label{rien}
e^{ng_{-}(e^{iu_{n}})} = e^{-\frac{n\ell}{2}}e^{\frac{ni}{2}(u_{n}+\pi)} \exp \left(-n \pi i \int_{u_{n}}^{\pi} d\mu(e^{i\alpha})\right).
\end{equation}
Substituting \eqref{rien} and \eqref{phi for sine kernel} in \eqref{Kernel in sine case}, we have
\begin{equation}\label{sine kernel0}
\frac{1}{cn}K_{n}(e^{iu_{n}},e^{iv_{n}}) = \frac{e^{\frac{ni}{2}(u_{n}-v_{n})}}{\pi P_{12}^{(\infty)}(0)(u-v)} \mbox{Im} \left( e^{\alpha_{n}} \widetilde K(e^{iu_{n}}) \overline{\widetilde K(e^{iv_{n}})} \right)  + \bigO(n^{-1}),
\end{equation}
as $n\to\infty$,
where $\displaystyle \alpha_{n} = \frac{n}{2}(\phi_{-}(e^{iu_{n}})-\phi_{-}(e^{iv_{n}})) = n\pi i \int_{v_{n}}^{u_{n}} d\mu(e^{i\theta})$. An explicit computation shows that
\begin{multline*}
\widetilde K(e^{iu_{n}}) \overline{\widetilde K(e^{iv_{n}})} = e^{-\alpha_{n}} \left( e^{\alpha_{n}} P_{11,-}^{(\infty)}(e^{iu_{n}}) \overline{P_{11,-}^{(\infty)}(e^{iv_{n}})} +e^{-\alpha_{n}}P_{12,-}^{(\infty)}(e^{iu_{n}}) \overline{P_{12,-}^{(\infty)}(e^{iv_{n}})} \right. \\
\left.  - e^{\beta_{n}}e^{-n\pi i \Omega} P_{11,-}^{(\infty)}(e^{iu_{n}}) \overline{P_{12,-}^{(\infty)}(e^{iv_{n}})} - e^{-\beta_{n}}e^{n\pi i \Omega } P_{12,-}^{(\infty)}(e^{iu_{n}}) \overline{P_{11,-}^{(\infty)}(e^{iv_{n}})} \right),
\end{multline*}
where $\beta_{n} = \frac{n}{2} \left( \phi_{-}(e^{iu_{n}})+\phi_{-}(e^{iv_{n}}) \right)$. If we choose $c = \psi_{x,L}(e^{iw})$, then $e^{\alpha_{n}} = e^{i\pi(u-v)}(1+\bigO(n^{-1}))$ as $n\to\infty$ and we obtain finally
\begin{align*}\mbox{Im} \left( e^{\alpha_{n}} \widetilde K(e^{iu_{n}}) \overline{\widetilde K(e^{iv_{n}})} \right) & =  \mbox{Im} \left( e^{\pi i (u-v)}|P_{11,-}^{(\infty)}(e^{iw})|^{2} + e^{-\pi i (u-v)}|P_{12,-}^{(\infty)}(e^{iw})|^{2} \right)\\&\hspace{5cm} +\bigO((u-v)n^{-1}), \\
& = \sin(\pi(u-v))P_{12}^{(\infty)}(0)+\bigO((u-v)n^{-1}),  
\end{align*}
as $n\to\infty$,
where we have again used Proposition \ref{prop Pinf relation}. Using this in \eqref{sine kernel0}, we obtain 
\begin{equation}
\lim_{n\to\infty} \frac{1}{cn} K_{n} \left( e^{i\left( w+\frac{u}{cn} \right)},e^{i\left( w+\frac{v}{cn} \right)} \right) = e^{\frac{i(u-v)}{2c}} \frac{\sin \pi(u-v)}{\pi(u-v)},
\end{equation}
which corresponds to \eqref{sine kernel limit} with $c=\psi_{x,L}(e^{iw})$.

\subsection{Cases I-III}\label{section: proof zeros other}

\subsubsection*{Zeros of $\phi_n$}
In Case III, it was shown in \cite{ChCl} that for any $\epsilon > 0$ fixed, if $|z| \geq 1+\epsilon$ or $|z| \leq 1-\epsilon$ we have
\[
\phi_{n}(z) = e^{n \int_{\gamma} \log(z-e^{i\theta})d\mu_{\infty,L}(e^{i\theta})} \frac{1}{2} \left( \left( \frac{z-\overline{z_{0}}}{z-z_{0}} \right)^{1/4}+\left( \frac{z-\overline{z_{0}}}{z-z_{0}} \right)^{-1/4} \right)\left( 1 + \bigO(n^{-1}) \right),
\]
as $n\to\infty$,
where the branch cuts are chosen on $\gamma$. In  a similar way as in Section \ref{section: proof zeros IV}, one has
\begin{equation}
\lim_{n\to \infty} \frac{1}{n} \mathrm{Re}(\log \phi_{n}(z))  = \int_{\supp (\mu_{\infty,L})} \log |z-e^{i\theta}| d\mu_{\infty,L}(e^{i\theta}), \qquad |z|\neq 1,
\end{equation}
and one shows exactly as in Section \ref{section: proof zeros IV} that
$\nu_{n,s,L}$ converges weakly to $\mu_{\infty,L}$.
In Case I, one can use \cite[formulas (4.1), (4.7), (4.56)]{DIK} to prove that
\begin{equation}
\lim_{n\to \infty} \frac{1}{n} \mbox{Re}(\log \phi_{n}(z))  =\begin{cases}\log |z|,&|z|>1\\
0,&|z|<1\end{cases}= \frac{1}{2\pi}\int_{0}^{2\pi} \log |z-e^{i\theta}| d\theta.
\end{equation}
The same holds in Case II by results from \cite{ClKr}.
This implies, by similar arguments as in Section \ref{section: proof zeros IV}, that $\nu_{n,s,L}$ converges weakly to the uniform measure on the circle.

Similar to Case IV, we have the following stronger result in Case I and in Case II, illustrated in Figure \ref{fig: zeros OP}.
\begin{figure}[t] \includegraphics[width=1\textwidth]{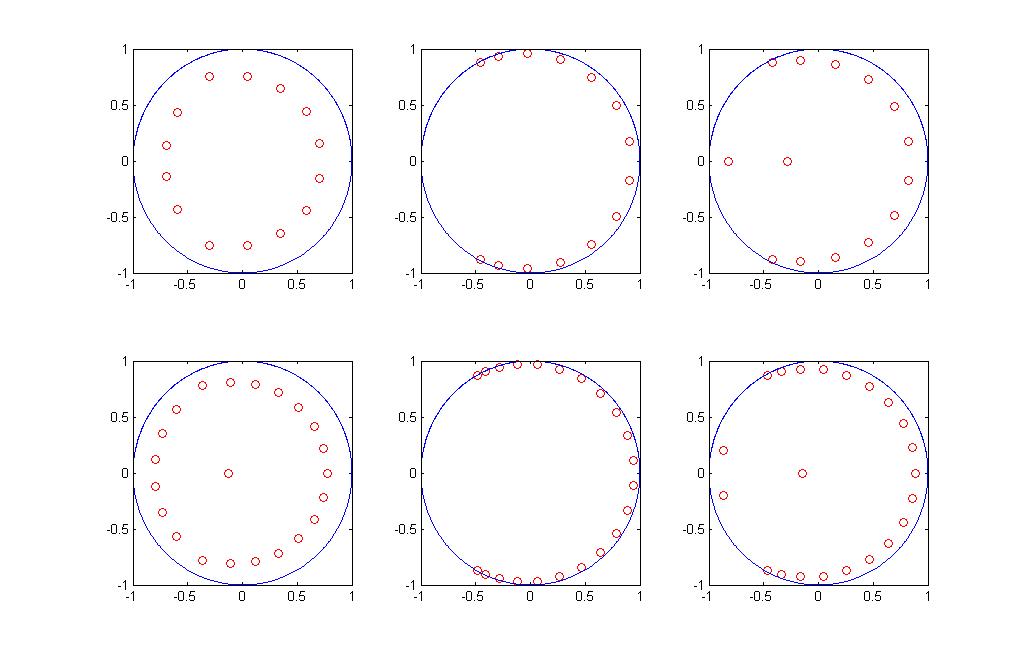} 
\caption{\label{fig: zeros OP} Zeros of $\phi_{14}(z)$ on the first row and of $\phi_{22}(z)$ on the second one, with $L = \frac{2\pi}{3}$. The first column corresponds to Case I with $s =1/2$, the second column to Case III with $x = 2x_{c}$ and the third column to Case IV with $x = x_{c}/2$.}
\end{figure}  
\begin{proposition} \label{prop: zeros phin2}
\begin{enumerate}
\item Fix $0<s<1$ and $0<L<\pi$. For any $\epsilon>0$, there exists $n_{0} \in \mathbb{N}$ such that for $n\geq n_{0}$, $\phi_{n}(z)$ has at most one zero in $\{z\in\mathbb C: |z| \leq 1 -\epsilon\}$. \item Fix $0<L<\pi$ and let $s=e^{-xn}$, $x \geq x_{c}$. For any $\epsilon>0$, there exists $n_{0} \in \mathbb{N}$ such that for $n\geq n_{0}$, $\phi_{n}(z)$ has no zeros in $\{z\in\mathbb C: |z| \leq 1 -\epsilon\}$.
\end{enumerate}
\end{proposition}
\begin{proof}
\begin{enumerate}
\item In the first case, we have for $|z|\leq 1-\epsilon$ by \cite[formulas (4.1), (4.3), (4.7), (4.8), (4.54), (4.56), (4.71)]{DIK} that
\[
\phi_{n}(z) = - \exp \left( -\frac{\log s}{2\pi i} \int_{L}^{2\pi-L} \frac{1}{w-z}dw \right) \left( \frac{c}{z-z_{0}} + \frac{\overline{c}}{z-\overline{z_{0}}} \right) \frac{1}{n} \left( 1 + \bigO(n^{-1}) \right),
\]
as $n\to\infty$,
where 
\[ 
c = z_{0}^{n+1} \beta_{0} \frac{\Gamma(\beta_{0})}{\Gamma(-\beta_{0})} |z_{0}-\overline{z_{0}}|^{-2\beta_{0}}n^{-2\beta_{0}}, \qquad \beta_{0} = \frac{-1}{2\pi i }\log s.
\]
Since this expression has just one root for sufficiently large $n$, the conclusion follows.
\item Here, we can use \cite[formulas (4.5), (4.13), (4.19), (4.36), (4.37)]{ChCl} as $n\to\infty$ for $|z|\leq 1-\epsilon$:
\[
\phi_{n}(z) = e^{n \int \log(z-e^{i\theta})d\mu_{\infty,L}(e^{i\theta})} \frac{1}{2} \left( \left( \frac{z-\overline{z_{0}}}{z-z_{0}} \right)^{1/4}+\left( \frac{z-\overline{z_{0}}}{z-z_{0}} \right)^{-1/4} \right)\left( 1 + \bigO(n^{-1}) \right),
\]
where the branch cut lies on $\gamma$. We easily verify that this expression vanishes nowhere for sufficiently large $n$.
\end{enumerate}
\end{proof}

\subsubsection*{Limiting mean eigenvalue density and correlation kernel}
Asymptotics for the mean eigenvalue density and for the correlation kernel can be obtained in Cases I-III using similar straightforward but lengthy calculations as in Section \ref{section: proof Bessel}, Section \ref{section: Airy}, and Section \ref{section: sine}, based on the asymptotic results in the papers \cite{DIK}, \cite{ClKr}, and \cite{ChCl}. In this way, one can prove for example formulas \eqref{psi s fixed}, \eqref{sine kernel limit}, \eqref{eq measure}, and \eqref{lim density}. For the sake of brevity, we do not present details here.

\appendix
\section{Proof of identity \eqref{identity intOmega}}
\label{appendix: proofOmega}

Recall the definition of the equilibrium measure $\mu_{x,L}$ in \eqref{def muxL}, and the related quantities $\Omega_{x,L}=\int_{\pi-T}^{\pi+T}d\mu(e^{i\theta})$ given in \eqref{def Omega} and $\ell_{x,L}$ given in \eqref{def ell}. We will use the identity $\ell_{\infty,L} = -2\log\sin \frac{L}{2}$ (see \cite{ChCl} for a proof of it).

\begin{proposition}\label{proposition for Omega}
Let $0<L<\pi$ and $0<x<x_c$. There holds a relation between $\Omega_{x,L}$, the Euler-Lagrange constant $\ell_{x,L}$, and $x$:
\begin{equation} \label{Equation for Omega}
\Omega_{x,L} = \frac{\partial \ell_{x,L} }{\partial x} + x \frac{\partial \Omega_{x,L}}{\partial x} = - \frac{\partial}{\partial x} \left[ -\frac{\ell_{\infty,L}}{2} \left( \frac{x}{x_{c}} \right)^{2} - \int_{\frac{x}{x_c}}^{1} \zeta \ell_{\frac{x}{\zeta},L}d\zeta \right].
\end{equation}
\end{proposition}
\begin{proof}
We consider the following meromorphic differential of the third kind on the Riemann surface $X$,
\begin{equation}
\widehat{\omega} = \sqrt{\frac{(z-z_{1})(z-\overline{z_{1}})}{(z-z_{0})(z-\overline{z_{0}})}} \frac{dz}{iz}.
\end{equation}
By the Riemann bilinear identity we have
\begin{equation} \label{Riemann_bilinear}
\int_{X} \omega \wedge \widehat{\omega} = A_{1}\widehat{B}_{1} - \widehat{A}_{1}B_{1},
\end{equation}
where $\omega$ is given by \eqref{def omega}, and
\begin{align*}
&A_{1} = \int_{A} \omega = 1, &&B_{1} =  \int_{B} \omega = \tau,\\
&\widehat{A}_{1} = \int_{A} \widehat{\omega} = -2ix,&& \widehat{B}_{1} = \int_{B} \widehat{\omega} = -4\pi \Omega_{x,L}.
\end{align*}

As $z \to 0$ on the first (respectively, second) sheet of $X$, we have
\begin{equation}
\widehat{\omega} = \pm \frac{i}{z} (1+\bigO(z)) dz, \qquad \omega = \mp c_{0}(1+\bigO(z))dz.
\end{equation}
In terms of the local coordinate $\xi = z^{-1}$, as $\xi \to 0$ on the first (respectively, second) sheet of $X$, we have
\begin{equation}
\widehat{\omega} = \pm \frac{i}{\xi} (1+\bigO(\xi)) d\xi, \qquad \omega = \mp c_{0}(1+\bigO(\xi))d\xi.
\end{equation}
The left hand side of \eqref{Riemann_bilinear} can now be computed using Stokes' theorem and the residue theorem. 
After a long calculation, we get $\int_{X} \omega \wedge \widehat{\omega} = -8\pi \int_{1^{(1)}_{-}}^{\infty^{(1)}} \omega$, where the notation $1_{-}^{(1)}$ means that we start the integration path from $1$ on the first sheet at the side $|z| > 1$. We obtain
\begin{equation}
\Omega_{x,L} = 2 \int_{1^{(1)}_{-}}^{\infty^{(1)}} \omega + \frac{i\tau}{2\pi}x.
\end{equation}
Now a direct computation shows that
\begin{equation}
\frac{\partial \Omega_{x,L}}{\partial x} = \frac{\partial \Omega_{x,L}}{\partial T} \frac{\partial T}{\partial x} = \frac{i\tau}{2\pi} \quad \mbox{ and } \quad \frac{\partial \ell_{x,L}}{\partial x} = \frac{\partial \ell_{x,L}}{\partial T} \frac{\partial T}{\partial x} = 2 \int_{1^{(1)}_{-}}^{\infty^{(1)}} \omega.
\end{equation}
This proves the first equality in \eqref{Equation for Omega}. Now, note that
\[
- \frac{\partial}{\partial x} \left[ -\frac{\ell_{\infty,L}}{2} \left( \frac{x}{x_{c}} \right)^{2} - \int_{\frac{x}{x_c}}^{1} \zeta \ell_{\frac{x}{\zeta},L}d\zeta \right] = \int_{\frac{x}{x_{c}}}^{1} \left.\frac{\partial \ell_{u,L}}{\partial u}\right|_{u = \frac{x}{\zeta}} d\zeta = x \int_{x}^{x_{c}} \frac{\Omega_{u,L} - u \frac{\partial \Omega_{u,L}}{\partial u} }{u^{2}}du,
\]
where we proceeded by substitution and used the first equality in \eqref{Equation for Omega}. Integrating by parts, we obtain
\[
\int_{x}^{x_{c}} \frac{\Omega_{u,L}}{u^{2}}du = \frac{\Omega_{x,L}}{x} + \int_{x}^{x_{c}} \frac{ 1}{u}\frac{\partial \Omega_{u,L}}{\partial u}du.
\]
This proves the second equality.
\end{proof}

\begin{corollary}
\begin{equation}
- \int_{0}^{x_{c}} \Omega_{x,L} dx = \frac{-\ell_{\infty,L}}{2} = \log \sin \frac{L}{2}
\end{equation}
\end{corollary}
\begin{proof}
It suffices to integrate \eqref{Equation for Omega} from $0$ to $x_{c}$ and to note that $\ell_{0,L} = 0$.
\end{proof}

\section{Extension of the RH analysis as $x\to 0$: shrinking Airy and Bessel disks}
\label{appendix: shrinking disks}

In this section we extend the RH analysis, up to now justified for $x$ in compact subsets of $(0,x_c)$, to the case of small $x$. This is the intermediate regime between Case IV and Case I. We first have a closer look at the behaviour of $T=T(x)$ as $x\to 0$.

\begin{proposition}\label{theta1_as_x_to_0}
As $\alpha := \pi - T - L \to 0_{+}$, we have
\begin{equation}\label{eq lambda}
\begin{array}{r c l}
\displaystyle x & = & \displaystyle \frac{\pi}{2} \alpha + \bigO(\alpha^{2}), \\[0.2cm]
\displaystyle \tau & = & - \displaystyle \frac{2i}{\pi} \log(\alpha) + \bigO(1).
\end{array}
\end{equation}
\end{proposition}
\begin{proof}
It was shown in \cite[proof of Proposition 3.1]{ChCl} that
\begin{equation}
x = \int_{L}^{L+\alpha} \sqrt{\frac{\cos \theta - \cos (L + \alpha)}{\cos L - \cos \theta}}d\theta.
\end{equation}
The equation for $x$ in \eqref{eq lambda} is therefore straightforward, since for any $a < b$, 
\begin{equation}
\int_{a}^{b} \sqrt{\frac{y-a}{b-y}}dy = \frac{\pi}{2} (b-a).
\end{equation}
To prove the second one, one starts from \eqref{tau def}, which gives
\begin{equation} \label{tau in alpha}
\tau = i\frac{\displaystyle \int_{L+\alpha}^{2\pi - (L+\alpha)} \frac{d\theta}{\sqrt{\cos L - \cos \theta} \sqrt{\cos (L+\alpha) - \cos \theta}} }{\displaystyle \int_{L}^{L+\alpha} \frac{d\theta}{\sqrt{\cos L - \cos \theta} \sqrt{\cos \theta - \cos (L+\alpha)}}}.
\end{equation} 
Since for any $a < b < c$, 
\[
\int_{a}^{b} \frac{dy}{\sqrt{b-y} \sqrt{y-a}} = \pi \] and\[\int_{a}^{b} \frac{dy}{\sqrt{b-y}\sqrt{c-y}} = - \log(c-b) + 2 \log(\sqrt{b-a}+\sqrt{c-a}),
\]
the denominator in \eqref{tau in alpha} is $ \frac{\pi}{\sin L} + \bigO(\alpha)$  and the numerator is $ \frac{-2}{\sin L} \log(\alpha) + \bigO(1)$ as $\alpha \to 0_{+}$, and the result follows.
\end{proof}

If we choose the radii of the disks of the four local parametrices to be, say, $r = \frac{x}{2\pi}$, it follows from Proposition \ref{theta1_as_x_to_0} that the disks do not intersect if $x$ is sufficiently small.

We now estimate the jump matrices for $R$ as $x\to 0$.
\begin{proposition}\label{R_as_x_to_0}
We let $r= \frac{x}{2\pi}$, and we let $R$ be as constructed in Section \ref{section: R}. As $n\to\infty$ and $x\to 0$ in such a way that $xn\to\infty$, we have
\begin{align}
&\label{jump R1} R_{+}(z) = R_{-}(z) \left(I + \bigO(x^{-1}n^{-1})\right), & \mbox{ for } z \mbox{ on $\partial D(\zeta,r)$, $\zeta=z_0,\overline{z_0}, z_1, \overline{z_1}$}, \\[0.2cm]
&\label{jump R2} R_{+}(z) = R_{-}(z) (I+\bigO(e^{-cxn})), & \mbox{ for } z \mbox{ elsewhere on } \Sigma_{R},
\end{align}
and $c > 0$ is a constant independent of $x$ and $n$.
\end{proposition}
\begin{proof}
By \eqref{jump R circles} and \eqref{P local para Bessel}--\eqref{def P Airy}, the jump matrices for $R$ are 
\begin{equation}\label{jump Rbis}
\begin{array}{l l}
I + P^{(\infty)}(z) \bigO \left( \frac{1}{n \phi(z)} \right) P^{(\infty)}(z)^{-1}, & \mbox{ for } z \in \partial D(z_{0},r),   \\
I + P^{(\infty)}(z) \bigO \left( \frac{1}{n \widetilde{\phi}(z)} \right) P^{(\infty)}(z)^{-1}, & \mbox{ for } z \in \partial D(\overline{z_{1}},r), \\
\end{array}
\end{equation}
and similarly on $\partial D(\overline{z_{0}},r)$ and $\partial D(z_{1},r)$. Uniformly for $z \in \mathcal{D}_{x} := \bigcup_{\zeta=z_{0},\overline{z_{0}},z_{1},\overline{z_{1}}} \partial D(\zeta,r)$, we have from \eqref{def_of_phi} and \eqref{def_of_tilde_phi}, as $x \to 0$, that there exists $\epsilon > 0$ such that $|\frac{\phi(z)}{r}| \geq \epsilon > 0$ and $|\frac{\widetilde{\phi}(z)}{r}| \geq \epsilon > 0$. To obtain \eqref{jump R1} from \eqref{jump Rbis}, it is now sufficient to prove that $
P^{(\infty)}(z)$ is uniformly bounded for $n$ large, $x$ small and for $z \in \Sigma_{R}$.

For $z$ outside the four disks, since the zero of $\theta(u(z)-d)$ is cancelled by the one of $\beta(z)-\beta^{-1}(z)$, it is straightforward to see that $P^{(\infty)}(z) = \bigO(1)$ outside of a neighbourhood of the four edge points, as $n\to \infty$ and $x\to 0$. On $\partial D(z_{0},r)$, from \eqref{equation for beta} and \eqref{eq lambda}, we have 
\begin{equation}
|\beta(z)|^{4} \leq 2 \frac{\frac{2}{\pi} x + r}{r} = 10 \qquad \mbox{ and } \qquad |\beta(z)^{-1}|^{4} \leq 2 \frac{r}{\frac{2}{\pi} x - r} = \frac{2}{3}.
\end{equation}
A similar computation shows that $\beta$ and $\beta^{-1}$ are also bounded on the other disks.

By \eqref{final expression Pinf}, we still need to show that
\begin{equation}\label{limsupTheta}
\sup_{\substack{z \in \mathcal{D}_{x} \\ y \in \mathbb{R}}} \Theta_{kl}(z), \qquad \mbox{ with } 1 \leq k,l \leq 2,
\end{equation}
are uniformly bounded for $n$ large and $x$ small.
From Proposition \ref{theta1_as_x_to_0}, we have that $|\tau| \to \infty$ as $x \to 0$.
It is clear from \eqref{def theta} that
\begin{equation}
\lim_{\tau \to i\infty} \sup_{y \in \mathbb{R}} \left| \theta(y;\tau) - 1 \right| = 0,
\end{equation} and it follows that  $ \frac{\theta(0;\tau)}{\theta(n\Omega;\tau)} \to 1$ as $n \to \infty$, $x\to 0$.

A next important observation is that there exist $x_{0} > 0$ and $0<\epsilon<\frac{1}{4}$ such that for all $x \leq x_{0}$ and $z \in \mathcal{D}_{x}$ we have
\begin{multline}
\Im (u(z) \pm d) \in \left[-\left( \frac{1}{4} + \epsilon \right) |\tau|,-\left( \frac{1}{4} - \epsilon \right) |\tau| \right] \cup \left[ \left( \frac{1}{4} - \epsilon \right) |\tau|, \left( \frac{1}{4} + \epsilon \right) |\tau| \right] \\ \cup \left[ \left( \frac{3}{4} - \epsilon \right) |\tau|, \left( \frac{3}{4} + \epsilon \right) |\tau| \right].
\end{multline}
From the periodicity properties  \eqref{periodicity of theta} of $\theta$, we have also that for any $y \in \mathbb{R}$,
\begin{equation}
\frac{|\theta(u(z) - d + y)|}{|\theta(u(z) - d)|} = \frac{|\theta(u(z) - d + y - \tau)|}{|\theta(u(z) - d - \tau)|}.
\end{equation}
Combining these two facts with the following estimate for any $0 \leq \delta < 1$,
\begin{equation}
\sup_{-\delta \frac{|\tau|}{2} \leq \Im (z) \leq \delta \frac{|\tau|}{2}} |\theta(z;\tau) - 1|  \leq 2 e^{-\pi |\tau| (1-\delta)} + \frac{2 e^{-2\pi |\tau|}}{1-e^{-\pi |\tau|}} \to 0, \qquad \mbox{ as } \tau \to i\infty,
\end{equation} we obtain that \eqref{limsupTheta} is bounded, which leads the first estimate in \eqref{jump R1}. By \eqref{jump R lenses}, we similarly obtain \eqref{jump R2}.
\end{proof}

If we choose $x$ such that $x\to 0$ and $x n \to \infty$ as $n \to \infty$, we thus have that $R$ exists for sufficiently large $n$ and satisfies
\begin{equation}
R(z) = I + \bigO(x^{-1}n^{-1}),\qquad R^{\prime}(z) = \bigO(x^{-1}n^{-1}).  \label{as R general}
\end{equation}
This allows one to generalize Theorem \ref{theorem: Toeplitz} and Theorem \ref{theorem macro} to the case where $x\to 0$ and $n\to\infty$ in such a way that $xn\to\infty$.
From these asymptotics for $R$, we can obtain asymptotics for the orthogonal polynomials $\phi_n$, for the Toeplitz determinants, the mean eigenvalue density, and the eigenvalue correlation kernel. In particular, as $x\to 0$ and $n\to\infty$ in such a way that $xn\to\infty$, we have that the zero counting measure $\nu_{n,s,L}$ of $\phi_n$ and the mean eigenvalue distribution $\mu_{n,s,L}$ converge weakly to the uniform measure on the unit circle. For the eigenvalue correlation kernel $K_n$, we have convergence to the sine kernel everywhere except near the points $e^{\pm i L}$.

\section*{Acknowledgements}
CC was supported by FNRS, TC was supported by the European Research Council under the European Union's Seventh Framework Programme (FP/2007/2013)/ ERC Grant Agreement n.\, 307074. Both authors also acknowledge support by the Belgian Interuniversity Attraction Pole P07/18.

\end{document}